\RequirePackage[backend=biber]{biblatex}

\newcommand\onlysubm[1]{} \newcommand\onlyarxiv[1]{#1} \newcommand\exceptarxiv[1]{} \newcommand\exceptsubm[1]{#1}

\exceptsubm{

\documentclass [ a4paper
, autoref
, english
] {lipics-v2021} 
\raggedbottom

\newcommand\onlylipics[1]{#1}
\newcommand\exceptlipics[1]{}
\newcommand\onlyllncs[1]{}
\newcommand\exceptllncs[1]{#1}
\newcommand\onlynone[1]{}

 }
\onlysubm{

\documentclass{llncs}
\raggedbottom

\newcommand\onlylipics[1]{}
\newcommand\exceptlipics[1]{#1}
\newcommand\onlyllncs[1]{#1}
\newcommand\exceptllncs[1]{}
\newcommand\onlynone[1]{}

 }

\usepackage[utf8]{inputenc}
\usepackage[T1]{fontenc}
\usepackage{csquotes}
\usepackage{dsfont,interval,blindtext,comment,braket}
\usepackage{float,tabulary,tabularx,multirow,multicol}
\usepackage{appendix,ifthen,changepage,letltxmacro}

\usepackage{amsmath,amssymb,mathtools}

\providecommand\ifllncs[2]{\onlyllncs{#1}\exceptllncs{#2}}

\providecommand\ifsubm[2]{\onlysubm{#1}\exceptsubm{#2}}

\exceptlipics{
	\usepackage{subfig}
	\usepackage{enumerate}
}

\makeatletter
\def\cleartheorem#1{\expandafter\let\csname#1\endcsname\relax
	\expandafter\let\csname c@#1\endcsname\relax
}

\exceptllncs{
	\usepackage{amsthm,thmtools}
	\makeatother

	\theoremstyle{definition}
	\newtheorem{thm}{Theorem}

	\cleartheorem{lemma}\newtheorem{lemma}[thm]{Lemma}
	\cleartheorem{propo}\newtheorem{propo}[thm]{Proposition}
	\cleartheorem{prob} 
	\cleartheorem{defi} 
	\cleartheorem{cor}  

	\cleartheorem{proposition}\newtheorem{proposition}[thm]{Proposition}
	\cleartheorem{problem}    
	\cleartheorem{definition} \newtheorem{definition} [thm]{Definition}
	\cleartheorem{corollary}  \newtheorem{corollary} [thm]{Corollary}
}

\onlyllncs{
	
	\spnewtheorem{thm}{Theorem}{\bfseries}{\rmfamily}

	\cleartheorem{theorem}
	\spnewtheorem{theorem}[thm]{Theorem}{\bfseries}{\rmfamily}

	\cleartheorem{lemma}
	\spnewtheorem{lemma}[thm]{Lemma}{\bfseries}{\rmfamily}

	\cleartheorem{propo}
	\spnewtheorem{propo}[theorem]{Proposition}{\bfseries}{\rmfamily}

	\cleartheorem{corollary}
	\spnewtheorem{corollary}[theorem]{Corollary}{\bfseries}{\rmfamily}

	\cleartheorem{proposition}
	\spnewtheorem{proposition}[theorem]{Proposition}{\bfseries}{\rmfamily}

	\cleartheorem{definition}
	\spnewtheorem{definition}[theorem]{Definition}{\bfseries}{\rmfamily}
}

\onlyllncs{
	\providecommand\qedhere{\tag*{\qed}}
	\providecommand\placeqed{\hfill\qed}
}
\providecommand\placeqed\relax

\usepackage{hyperref,xcolor}

\providecolor{Green}{HTML}{00a650}
\hypersetup{
	colorlinks   = true, urlcolor     = blue, linkcolor    = blue!50!red!78, citecolor   = Green }

\onlyllncs{
	\let\theGamma=\Gamma
	\let\thePhi=\Phi
	\let\theLambda=\Lambda
	\let\theOmega=\Omega
	\let\theTheta=\Theta

	\renewcommand\Gamma{\mathrm{\theGamma}}
	\renewcommand\Phi{\mathrm{\thePhi}}
	\renewcommand\Lambda{\mathrm{\theLambda}}
	\renewcommand\Omega{\mathrm{\theOmega}}
	\renewcommand\Theta{\mathrm{\theTheta}}
	\renewcommand\Upsilon{\mathrm{\Upsilon}}
	\renewcommand\Psi{\mathrm{\Psi}}
	\renewcommand\Xi{\mathrm{\Xi}}
	\renewcommand\Sigma{\mathrm{\Sigma}}
}
\usepackage{xparse}

\newcommand\wrap[2]{
    \expandafter\expandafter\expandafter\let\expandafter\expandafter\csname the#1\endcsname\csname #1\endcsname
    \expandafter\expandafter\expandafter\def\expandafter\expandafter\csname #1\endcsname\expandafter{
        \expandafter\let\expandafter\origin\csname the#1\endcsname
        #2
    }
}

\newcommand\map[3][ ]{
		\makeatletter
		\let\internal@map=\relax
		\NewDocumentCommand\internal@map{ m >{\SplitList{#1}} m }{ \ProcessList{##2}##1 }
		\internal@map#2{#3}
		\makeatother
}

\usepackage{xparse}
\newcommand\true{\mathrm{true}}
\newcommand\false{\mathrm{false}}
\newcommand\komma{,\,}
\NewDocumentCommand\overann{s O{\big} m m}{
	\IfBooleanTF{#1}{\overset{\mathclap{#4}}{#3}}{
		\overset{\mathclap{\substack{{#4} \\ #2\downarrow}}}{#3}}
}
\NewDocumentCommand\underann{s O{\big} m m}{
	\IfBooleanTF{#1}{\overset{\mathclap{#4}}{#3}}{
		\underset{\mathclap{\substack{ #2\uparrow \\ {#4}}}}{#3}}
}

\DeclareRobustCommand{\stopp}{\leavevmode\unskip\penalty9999 \hbox{}\nobreak\hfill
	\quad\hbox{\ensuremath{\lrcorner}}}
\providecommand\stopphier{\tag*{\ensuremath{\lrcorner}}}

\providecommand\localqedhere{\tag*{\ensuremath{\vartriangleleft}}}

\NewDocumentCommand{\semi}{t& t| t. O{;}}{\IfBooleanTF#3{.}{#4\IfBooleanT{#1}{\text{\ and\ }}\IfBooleanT{#2}{\text{\ or\ }}}}

\NewDocumentCommand\dpm{s m}{
	{\,\hfill$\displaystyle{#2}$\hfill\,}\IfBooleanTF{#1}{\linebreak}{}
}

\newcommand\sing[1]{\{#1\}}
\newcommand\Sing[1]{\left\{#1\right\}}

\newcommand\N{\ensuremath{\mathds N}}

\newcommand\Q{\ensuremath{\mathds Q}}
\DeclareMathOperator{\poly}{poly}

\DeclarePairedDelimiter\abs{\lvert}{\rvert}\DeclarePairedDelimiter\ceil{\lceil}{\rceil}

\let\O=\relax
\DeclareMathOperator{\O}{\mathcal O}

\newcommand\defeq\coloneqq
\newcommand\eps\varepsilon
\makeatletter
	\newcommand{\dotminus}{\mathbin{\text{\@dotminus}}}
	\newcommand{\@dotminus}{\ooalign{\hidewidth\raise1ex\hbox{.}\hidewidth\cr$\m@th-$\cr}}
	\makeatother

\DeclarePairedDelimiter{\internalXenpar}{\lparen}{\rparen}
\NewDocumentCommand{\enpar}{s t~ m}{
	\IfBooleanTF{#2}{{\IfBooleanTF{#1}{\internalXenpar*{#3}}{\internalXenpar{#3}}}}{\IfBooleanTF{#1}{\internalXenpar*{#3}}{\internalXenpar{#3}}}}

\newcommand\totherwise{\text{otherwise}}

\newcommand\page[1]{p.\@ #1}

\newcommand\Lemma[1]{Lemma #1}

\RequirePackage{etoolbox} \RequirePackage{lineno}
\nolinenumbers
\newcommand*\linenomathpatch[1]{\expandafter\pretocmd\csname #1\endcsname {\linenomath}{}{}\expandafter\pretocmd\csname #1*\endcsname{\linenomath}{}{}\expandafter\apptocmd\csname end#1\endcsname {\endlinenomath}{}{}\expandafter\apptocmd\csname end#1*\endcsname{\endlinenomath}{}{}}
\newcommand*\linenomathpatchAMS[1]{\expandafter\pretocmd\csname #1\endcsname {\linenomathAMS}{}{}\expandafter\pretocmd\csname #1*\endcsname{\linenomathAMS}{}{}\expandafter\apptocmd\csname end#1\endcsname {\endlinenomath}{}{}\expandafter\apptocmd\csname end#1*\endcsname{\endlinenomath}{}{}}

\expandafter\ifx\linenomath\linenomathWithnumbers
  \let\linenomathAMS\linenomathWithnumbers
\patchcmd\linenomathAMS{\advance\postdisplaypenalty\linenopenalty}{}{}{}
\else
  \let\linenomathAMS\linenomathNonumbers
\fi

\linenomathpatch{equation} \linenomathpatchAMS{gather}
\linenomathpatchAMS{multline}
\linenomathpatchAMS{align}
\linenomathpatchAMS{alignat}
\linenomathpatchAMS{flalign}

\makeatletter

\newcommand\multiautoref[1]{\@first@ref#1,@}
\def\@throw@dot#1.#2@{#1}\def\@set@refname#1{\edef\@tmp{\getrefbykeydefault{#1}{anchor}{}}\xdef\@tmp{\expandafter\@throw@dot\@tmp.@}\ltx@IfUndefined{\@tmp autorefnameplural}{\def\@refname{\@nameuse{\@tmp autorefname}s}}{\def\@refname{\@nameuse{\@tmp autorefnameplural}}}}
\def\@first@ref#1,#2{\ifx#2@\autoref{#1}\let\@nextref\@gobble \else \@set@refname{#1}\@refname~\ref{#1}\let\@nextref\@next@ref \fi \@nextref#2}
\def\@next@ref#1,#2{\ifx#2@ and~\ref{#1}\let\@nextref\@gobble \else, \ref{#1}\fi \@nextref#2}
\makeatother

\newcounter{counter:resume}
\newcommand\suspend{\setcounter{counter:resume}{\value{enumi}}}
\newcommand\resume{\setcounter{enumi}{\value{counter:resume}}}

\def\midvert{\egroup : \bgroup}

\makeatletter
\def\mid@vertical{{:}}
\makeatother

\newcommand\B{{\Gamma}}

\newcommand\valid{\Lambda}
\newcommand\smin[1]{\setminus \sing#1}
\newcommand\partials[1]{\partial #1}
\newcommand\tup[2]{\Big( #1,\, #2 \Big)} 
\newcommand\tuplike[1]{\big( #1 \big)}
\newcommand\myqedhere{\qedhere\!\!\!\!\!\!}

\newcommand\omitted{\begin{proof}[Omitted, see appendix]\end{proof}}

\providecommand\mysec[1]{\noindent{\textbf{#1}\hspace{0.5em}}}
\providecommand\mysec[1]{\vspace{0.4em}\noindent\textbf{\large #1}.}
\newcommand\defequiv{\mathbin{\vcentcolon\equiv}}

\NewDocumentCommand{\itemref}{s m}{\IfBooleanTF#1{(\ref*{#2})}{\hyperref[#2]{(\ref*{#2})}}}

\renewcommand\subref[2]{\hyperref[#1:#2]{\autoref*{#1}~\itemref*{#1:#2}}}
\newcommand\stmtref[1]{\hyperref[#1]{Statement \itemref*{#1}}}
\newcommand\dref[1]{\text{\hyperref[#1]{(\ref*{#1})}}}
\newcommand\autodref[1]{\text{\autoref{#1}}}
\newcommand\dsubref[2]{\text{\hyperref[#1:#2]{(\ref*{#1}\ref*{#1:#2})}}}
\newcommand\autodsubref[2]{\text{\hyperref[#1:#2]{\autoref*{#1}\ref*{#1:#2}}}}

\newcommand\subpoint[1]{\vspace{0.1em}\noindent \exceptllncs{\textit{Proof of {#1}.}}\onlyllncs{\underline{Proof of {#1}.}}}

\title{Solving Cut-Problems in Quadratic Time for Graphs With Bounded Treewidth}

\onlylipics{
	\author{Hauke Brinkop}{Kiel University, Kiel, Germany}{brinkop.edu@gmail.com}{https://orcid.org/0000-0002-7791-2353}{}
	\author{Klaus Jansen}{Kiel University, Kiel, Germany}{kj@informatik.uni-kiel.de}{https://orcid.org/0000-0001-8358-6796}{}
}
\onlyllncs{
	\institute{Kiel University}
	\author{
		Hauke Brinkop
\orcidID{0000-0002-7791-2353}
		\and
		Klaus Jansen
\orcidID{0000-0001-8358-6796}
	}
}
\authorrunning{H.\@ Brinkop \and K.\@ Jansen} \onlylipics{\Copyright{Hauke Brinkop}
}

\onlylipics{\funding{This work was partially supported by DFG Project
\href{https://gepris.dfg.de/gepris/projekt/453769249}{\enquote{Fein-granulare Komplexität und Algorithmen für Scheduling und Packungen}}, \texttt{JA 612 /25-1}}}

\newcommand\submorarxiv[2]{\onlysubm{#1}\onlyarxiv{#2}}

\onlylipics\hideLIPIcs

\bibliography{cut.bib}

\AtEveryBibitem{ \clearname {editor} \clearfield {series} \clearfield {url} \clearfield {issn} }

\providecommand\para[1]{\paragraph*{#1}}
\providecommand\paras[1]{\paragraph*{#1}}
\newcommand\temptext{\exceptllncs{Proof}}
\let\tilde=\widetilde
\let\hat=\widehat
\let\sqleq=\sqsubseteq

\renewcommand\uplus{\mathbin{\dot\cup}}
\renewcommand\biguplus{\mathop{\dot\bigcup}}
\let\defequiv=\defeq \let\equiv==

\renewcommand\alpha{w} \newcommand\fullversion{extended version~\cite{fullversion}}
\onlysubm{\renewcommand\omitted{\begin{proof}[Omitted due to space restrictions, we refer to the \fullversion]\end{proof}}}

\begin{document} \maketitle
	\begin{abstract}\noindent
	In the problem (Unweighted) Max-Cut we are given a graph $G = (V,E)$ and asked
	for a set $S \subseteq V$ such that the number of edges from $S$ to $V \setminus S$
	is maximal. In this paper we consider an even harder problem: (Weighted) Max-Bisection. Here 
	we are given an undirected graph $G = (V,E)$ and a weight function $w \colon E \to \Q_{>0}$ 
	and the task is to find a set $S \subseteq V$ such that
	(i) the sum of the weights of edges from $S$ is maximal; and 
	(ii) $S$ contains  $\ceil*{\frac{n}{2}}$ vertices (where $n = \abs V$). We design
	a framework that allows to solve this problem in time $\O(2^t n^2)$ if 
	a tree decomposition of width $t$ is given as part of the input.
	This improves the previously best running time for Max-Bisection of \textcite{DBLP:journals/tcs/HanakaKS21} by a factor $t^2$.
	Under common
	hardness assumptions, neither the dependence on $t$ in the exponent nor the dependence on $n$
	can be reduced~\cite{DBLP:journals/tcs/HanakaKS21,DBLP:journals/jcss/EibenLM21,DBLP:journals/talg/LokshtanovMS18}.
	Our framework
	can be applied to other cut problems like 
	Min-Edge-Expansion, Sparsest-Cut, Densest-Cut, $\beta$-Balanced-Min-Cut,
	and Min-Bisection.
	It also works in the setting with arbitrary weights and directed edges. 
	\end{abstract}
	
	\section{Introduction}
	\label{sec:introduction}
	Unweighted Max-Cut is one of Karp's 21 NP-complete problems~\cite{Karp1972};
	given a graph $G = (V,E)$ one is asked for a set $S \subseteq V$ such that the number
	of edges from $S$ to $V \setminus S$ is maximal.
	Formally, a \emph{cut} is determined
	by a set of vertices $S \subseteq V$ of a graph. The \emph{size of a cut} is given by
	the number of edges from $S$ to $V \setminus S$. We denote these
	edges as $\partial S$ and, for the sake of shortness, if $S = \sing v$ for some $v$, 
	we write $\partials v$ instead of $\partial S = \partial \sing v$.
	If the graph is weighted,
	the size of the cut is given by the sum of the edge weights
	$\alpha(\partial S) \defeq \sum_{e \in \partial S} w(e)$ instead of their number
	$\abs{\partial S}$.
	In this paper we consider different cut problems for directed and weighted graphs\footnote{The undirected and unweighted versions can easily be modelled as directed
	and weighted by setting each edge weight to $1$ and by replacing each undirected edge
	between vertices $v_1$ and $v_2$ by two directed edges, $v_1v_2$ and $v_2v_1$.},
	more precisely Max-Cut, $\beta$-Balanced-Min-Cut, Max-Bisection, Min-Bisection, Min-Edge-Expansion, (uniform) Sparsest-Cut, and Densest Cut.
	See \autoref{table:objectives} for precise formulations of these
	problems.

	Observe that Densest-Cut and Sparsest-Cut can easily be reduced
	on each other in time $\O(n^2)$
	(where $n = \abs V$)
	using the complementary graph~\cite{DBLP:journals/jda/BonsmaBPP12}; however, this reduction might change the treewidth and the corresponding
	decomposition, hence we have to consider both problems individually.
	\begin{table}
		\centering
		\begin{tabulary}{\textwidth}{|LCC|} 
			\hline
			Name & Weights & Objective 
			\\\hline\hline
			\rule{0pt}{1.5em}Max-Cut&arbitrary& $\displaystyle \max_{S \subseteq V} \alpha(\partial S)$ \\ \rule{0pt}{2em}$\beta$-Balanced-Min-Cut~\cite{DBLP:journals/ipl/FeigeY03} & non-negative &$\displaystyle \min_{\substack{S \subseteq V \\ \beta \cdot \abs V \leq \abs S \leq (1 - \beta) \cdot \abs V}} \hskip -2em \alpha(\partial S)$ \\ 
			\rule{0pt}{2em}Max-Bisection~\cite{DBLP:journals/siamcomp/JansenKLS05}&non-negative& $\displaystyle  \max_{\substack{S \subseteq V \\ \abs[\big]{\abs S - \abs{V \setminus S}} \leq 1 }}
			\alpha(\partial S)$ \\ \rule{0pt}{2em}Min-Bisection~\cite{DBLP:journals/siamcomp/JansenKLS05}&non-negative& $\displaystyle \min_{\substack{S \subseteq V \\ \abs[\big]{\abs S - \abs{V \setminus S}} \leq 1 }}
			\alpha(\partial S)$\\ \rule{0pt}{2em}Min-Edge-Expansion~\cite{mahoney}&non-negative& $ \displaystyle \min_{\substack{\emptyset \neq S \subseteq V \\ \abs S \leq \abs{V \setminus S}}}
			\frac{\alpha(\partial S) }{\abs S} $ \\ \rule{0pt}{2em}Sparsest-Cut~\cite{DBLP:journals/jda/BonsmaBPP12}&non-negative& $\displaystyle  \min_{\emptyset \neq S \subsetneq V}
			\frac{\alpha(\partial S) }{\abs S \cdot \abs{V \setminus S}}$\\ \rule[-1.5em]{0pt}{3.5em}Densest-Cut~\cite{DBLP:journals/jda/BonsmaBPP12}&non-negative& $\displaystyle \max_{\emptyset \neq S \subsetneq V}
			\frac{\alpha(\partial S) }{\abs S \cdot \abs{V \setminus S}}$ \\ \hline
		\end{tabulary}
		\caption{Problems that we solve in quadratic time.}
		\label{table:objectives}
	\end{table} 
	We want to point out that if negative edge weights are allowed, as they are in our algorithm,
	Max-Bisection and Min-Bisection coincide. This does not hold for Max-Cut and its
	corresponding minimization variant; Min-Cut is solvable in polynomial time.
	We call those more general variants
	of the problem, where we get rid of the non-negativity restrictions,
	Max-Bisection', $\beta$-Balanced-Min-Cut', Min-Edge-Expansion', Sparsest-Cut', and Densest-Cut'. As we
	will see, our framework is able to solve these more general variants of the
	problems.

	Many graph problems are in FPT if parametrized by treewidth. This holds
	especially for the problems mentioned above~\cite{DBLP:journals/jda/BonsmaBPP12,DBLP:journals/tcs/HanakaKS21,DBLP:journals/jcss/EibenLM21,DBLP:journals/corr/abs-1910-12353}.
	The corresponding algorithms
	usually assume that a tree decomposition with $\O(n)$ (or a similar bound like $\O(nt)$)
	nodes 
	of width $t$ is given as part of the input;
	we will later see why that is a reasonable assumption and what we can do if a tree
	decomposition is not given.
	We show that in the above setting, all aforementioned problems can be solved in
	time $\O(2^tn^2)$.

	\section{Related Work} \textcite{DBLP:journals/siamcomp/JansenKLS05}
	proposed an algorithm for Max-Bisection and Min-Bisection that runs in time
	$\O(2^t n^3)$ for a graph with $n$ vertices, given a tree decomposition of width $t$ with $\O(n)$ nodes\footnote{The upper bound on the number of nodes occurs only
	implicitly in their work within the analysis of their algorithm's running time.}.
	They transform the tree decomposition into a so-called
	\emph{nice tree decomposition}~\cite{Kloks1994TreewidthCA}
	and then formulate a dynamic program over
	the nodes of the tree decomposition. The bottleneck of their analysis are nodes that have more than one child,
	the so-called \emph{join nodes}. There might be $\Omega(n)$ join nodes and for each
	the dynamic program might take time $\Omega(2^tn^2)$ to compute all the entries.
	\textcite{DBLP:journals/jcss/EibenLM21} have been able to improve the running time\footnote{To be precise, they consider Min-Bisection; however, in their paper as well
	as in~\cite{DBLP:journals/siamcomp/JansenKLS05}, all arguments work for both problems.}
	in its dependence on $n$ by balancing the tree decomposition and recognizing that
	the entries in join nodes can be computed via $(\max,+)$-convolution; this yields
	a running time of $\O(8^t t^5 n^2 \log n)$. \textcite{DBLP:journals/tcs/HanakaKS21} proved that
	the algorithm of \citeauthor{DBLP:journals/siamcomp/JansenKLS05} does in fact run
	in time $\O(2^t(nt)^2)$, using a clever idea to improve the analysis: while for a single
	join node computing all entries of the dynamic program might take $\Omega(2^tn^2)$, the 
	overall time for all join nodes altogether is $\O(2^t(nt)^2)$.

	\textcite{DBLP:journals/talg/LokshtanovMS18} proved that Max-Cut (without a tree decomposition given as part of the input) cannot be solved in time $\O^( (2-\eps)^t \poly n)$
	for any $\eps > 0$
	assuming the \emph{Strong Exponential Time Hypothesis (SETH)}~\cite{DBLP:journals/jcss/ImpagliazzoP01,DBLP:conf/iwpec/CalabroIP09}.
	It is not hard to see that this result can be extended to the case where a tree
	decomposition is part of the input -- for the sake of completeness, we include
	the corresponding proof in the appendix.
By adding isolated vertices, this result can also be applied to Max-Bisection
	and Min-Bisection~\cite{DBLP:journals/tcs/HanakaKS21}.
\textcite{DBLP:journals/jcss/EibenLM21} proved that Min-Bisection (and hence Max-Bisection')
	cannot be solved in truly subquadratic time, that is $\O(n^{2 - \eps})$ 
	for some $\eps > 0$, even if a tree decomposition of width $1$ is given as part of the input, unless
	$(\min,+)$-convolution can be solved in truly subquadratic time, which is considered
	unlikely~\cite{cygan_2017}. 

	Given a tree decomposition of width $t$ with $\O(nt)$ nodes, in time
	$\O(2^tn^3)$ the problems Sparsest-Cut~\cite{DBLP:journals/jda/BonsmaBPP12},
	Densest-Cut~\cite{DBLP:journals/jda/BonsmaBPP12}, and Min-Edge-Expansion~\cite{DBLP:journals/corr/abs-1910-12353} can be solved. To our knowledge,
	those are the best running times achieved so far.

	\para{Our Contribution \& Organization Of This Paper}
	In \autoref{sec:bisection} we show how to improve  the running time by a factor $t^2$ for Max-Bisection. In \autoref{sec:framework} we then generalize this to a framework which can be used
	to solve different cut problems in time $\O(2^tn^2)$ (compare \autoref{table:objectives}).
	Some problems (like Sparsest Cut) are improved by a factor $n$, which is substantial when $t$ is small.\ifsubm{
		The instantiations of our framework together with the corresponding correctness proofs
are ommited due to space restrictions. We refer to the \fullversion\ of the paper.
	}{
		In \autoref{sec:instantiations} we eventually
		show how the framework can be used to solve
		aforementioned problems, including the proofs
		that the corresponding models are correct.
	}

	\section{Preliminaries} \label{sec:prelim}

	\mysec{Notation} For tuples we write $(a,b) \oplus (c,d) = (a + c, b + d)$
	(and analogously define $\ominus$). We write $\pi$ to denote the projection on tuples,
	that is: for a tuple $t$ and an index $i$, $\pi_i(t)$ is the $i$-th component of $t$.

	For a set $M$ and a number $k$ we write
	$\binom M k = \set{M' \subseteq M | \abs M = k}$.
	For a graph edge
	from $v_1$ to $v_2$ we write $v_1v_2$ for both, directed and undirected
	graphs. A rooted tree $T = (V,r,E)$ is a graph $(V,E)$ that is connected, has no circles, and where $r \in V$. 

	We use $n$ as abbreviation for $\abs V$.
	By $\bar \Q \defeq \Q \uplus \sing{\infty, -\infty}$ we denote the rational
	numbers with positive and negative infinity.

	As already mentioned in the previous sections,
	we make use of \emph{tree decompositions}, which are defined as follows:
	\begin{definition}[Tree decomposition]
		\label{defi:tree-decomposition}
		Let $G = (V,E)$ be an undirected graph,
		$I$ be a finite set and $X = (X_i)_{i \in I}$ be a family
		of sets such that for any $i \in I$ one has $X_i \subseteq V$.
Moreover, let $T = (I, r, H)$ be a rooted tree with root $r \in I$.
		Then $(I, r, X, H)$ is called a
		\emph{tree decomposition} of $G$ if and only if $T$ has the
		following properties:
		\begin{enumerate}[(i)]
			\item \label{defi:tree-decomposition:node_coverage}
				\textbf{Node coverage}: Every vertex occurs in some $X_i$ for some $i \in I$
				(and no further vertices occur): $\bigcup_{i \in I} X_i = V$;
			\item \label{defi:tree-decomposition:edge_coverage}
				\textbf{Edge coverage}: for every edge $v_1v_2 \in E$ there is a node
				$i \in I$ such that both, $v_1$ and $v_2$, are contained in $X_i$\semi
\item \label{defi:tree-decomposition:coherence}
				\textbf{Coherence}: for every vertex $v \in V$ the subgraph $T - \set{ i \in I | v \notin X_i}$ is connected\semi.
		\stopp\end{enumerate}
	\end{definition}

	\noindent By convention
	the nodes of the graph $G$ are called
	\emph{vertices} while the nodes of the tree are
	just called \emph{nodes}. For a node $i \in I$
	the set
	$X_i$ is called a \emph{bag}. The \emph{width} of a decomposition is the largest cardinality of any of its bags minus 1. The minimum width among all decompositions
	of $G$ is called the \emph{treewidth} of $G$.
	If the node set of a decomposition
	is sufficiently small, more precisely if $\abs I \leq 4 \cdot (\abs V + 1)$, we
	call the decomposition \emph{small}.
	A tree decomposition of a directed graph is a tree decomposition of the
	underlying graph.

	As the approaches mentioned before, our approach also makes use of a specific kind
	of tree decompositions, which are of a very simple structure:
	\begin{definition}[Nice Tree Decomposition]
		Let $G = (V,E)$ be an undirected graph and
		$(I, r, X, H)$ a tree decomposition of $G$.
		We call $(I, r, X, H)$ a \emph{nice tree decomposition}
		if and only if
		for any $i \in I$ the node $i$ is of one of the following forms:
		\begin{enumerate}[(i)]
			\item \textbf{Leaf node}: $i$ has no child node in $T$,
			that is $i$ is
				a \emph{leaf} of the tree $T$\semi

			\item \textbf{Forget node}: $i$ has exactly one child node $j \in I$
				in $T$
				and $X_i \uplus \sing v = X_j$ for some $v \in X_j$, that is $i$
				\emph{forgets} a vertex from $X_j$\semi
					
			\item \textbf{Introduce node}: $i$ has exactly one child $j \in I$
				in $T$ and $X_i = X_j \uplus \sing v$ for some
				$v \in V \setminus X_j$, that is $i$
				\emph{introduces} a new vertex $v \in V \setminus X_j$\semi|

			\item \textbf{Join node}: $i$ has exactly two child nodes
				$j \in I$, $k \in I$, $j \neq k$,
				in $T$
				such that $X_i = X_j = X_k$,
				that is 
				\emph{joining} two branches of the tree $T$\semi.
			\stopp \end{enumerate} 
	\end{definition}

	\noindent The conversion of a tree decomposition into a nice tree decomposition
	can be done in time $\O(nt^2)$ as long as the
	number of nodes of the decomposition is at most linear in the 
	number of vertices.

	\begin{lemma}[{\cite[\Lemma{13.1.3}, \page{150}]{Kloks1994TreewidthCA}}]
		\label{stmt:nice}
		Given a small tree decomposition of a graph $G$ with width
		$t$
		one can
		find a nice tree decomposition of $G$ with width $t$ and
		with at most $4n$ nodes in $\O(nt^2)$ time,
		where $n$ is the number of vertices of $G$.
		\stopp
	\end{lemma}
	It is reasonable to assume that given tree decompositions are small for the following
	reason: No matter how the tree decomposition is constructed, it is always possible
	to incorporate the following mechanism without asymptotically increasing the running
	time:
If a node $j$ with parent $i$ has $X_j \subseteq X_i$, merge
		those nodes.

	We claim that we now can only have $n$ edges. This is because for every node $j$
	with parent $i$, we have $X_j \not \subseteq X_i$; this means, at least one
	vertex has to disappear when going from $j$ up to $i$. Since by \hyperref[defi:tree-decomposition:coherence]{Coherence} every
	vertex can disappear at most once\footnote{
		It might disappear on multiple leaf-root-paths; however, the node at which
		a vertex disappears, is the same on each of those leaf-root-paths.
	}, this upper bounds the number of edges by $n$ and
	hence the number of nodes by $n+1$.

	Note that this also allows us to easily extend our approach to the case where
	a tree decomposition with $\O(nt)$ nodes is given as part of the input: simply
	apply the procedure described above to reduce the number of nodes.

\section{Max-Bisection: From $\O(2^t n^3)$ to
	$\O(2^tn^2)$} \label{sec:bisection}
	In this section we focus on our idea on how the running time of the algorithm
	of \textcite{DBLP:journals/siamcomp/JansenKLS05} for Max-Bisection can be improved
	to $\O(2^tn^2)$, incorporating the idea of \textcite{DBLP:journals/tcs/HanakaKS21}.

	The algorithm of \textcite{DBLP:journals/siamcomp/JansenKLS05} is a subroutine
	used in their PTAS for the Max-Bisection problem on planar graphs. Their approach
	uses Baker's technique (see \cite{10.1145/174644.174650}) where the idea is
	to solve the problem for $k$-outerplanar graphs (instead of general planar graphs),
	for a $k$ depending only on
	the approximation factor, and then combining the results.
	Note that $k$-outerplanar graphs have a treewidth in $\O(k)$~\cite{DBLP:journals/corr/abs-1301-5896}.
	For those $k$-outerplanar graphs, the problem is solved exactly using the
	aforementioned subroutine. Since -- as opposed to the 
	general case -- tree decompositions for $k$-outerplanar graphs can be computed
	in time $\O(kn)$~\cite{DBLP:journals/corr/abs-1301-5896}, this
	subroutine gets tree decomposition as part of its input; otherwise the running
	time of the subroutine
	would be dominated by the computation of the decomposition.
	
	Let us now focus on the subroutine.
	We will traverse the nice tree decomposition bottom up in the algorithm
	of \citeauthor{DBLP:journals/siamcomp/JansenKLS05}, hence the following notations come
	in handy: For a node $i$ the set $Y_i$ contains all the vertices appearing in bags associated with nodes below $i$.
Moreover, we write $F_i \defeq Y_i \setminus X_i$
	to describe the set of vertices that \enquote{have been forgotten}
	somewhere	below $i$, that is, that they have appeared in bag of some node $j$ below $i$,
	but are not contained in $X_i$. Due to \subref{defi:tree-decomposition}{coherence}, those vertices can never reoccur
	in any bag above $i$.

	\let\oldGamma=\Gamma \renewcommand\Gamma{\mathcal B}
	The algorithm of \citeauthor{DBLP:journals/siamcomp/JansenKLS05} uses
	a dynamic program to compute
	\begin{align} \label{eq}
		& \Gamma_i \colon \set{0, \dots, \abs {Y_i}} \times 2^{X_i} \to \bar\Q_{> 0} \notag
		\\
		& \Gamma_i(\ell, S) = \max_{\substack{
			\hat S \subseteq Y_i 
			\\
			\abs{\hat S} = \ell
			\\
			S \subseteq \hat S
		}} \alpha(\partial \hat S \cap Y_i^2),
	\end{align}
	given a small nice tree decomposition of a weighted, undirected graph $G = (V,E)$
	with weight function $w \colon E \to \Q_{> 0}$. The idea is that for a node $i$
	the entry $\Gamma_i(\ell, S)$ is the size of the largest possible cut that consists
	of $\ell$ vertices from $Y_i$ and includes the set $S \subseteq X_i$.
	As the table might have preimages $(\ell, S)$ where there does not exist be a cut meeting
	the requirements above, $\infty$ and $-\infty$ have to be used to deal with
	those -- we omit further details. 
	
	For the root $r$ of the tree decomposition we can compute our optimal objective
	value by iterating over all entries of $\Gamma_r$ and picking the best value where
	the number of
	vertices is in the feasible interval for Max-Bisection.

	The dynamic program traverses
	the decomposition bottom up. The bottleneck of the running time comes from the time
	spent at join nodes -- the values for each node of any different type can be computed
	in time $\O(2^tn)$ using a simple DP.
	For a join node $i$ with left child $j$ and right child $k$,
	they use the following recurrence to compute $\Gamma_i$:
	\begin{equation} \label{eq:join_jansen}
		\Gamma_i(\ell, S) = \max_{\substack{
			\abs S \leq \ell_1 \leq \abs  {V} \\
			\abs S \leq \ell_2 \leq \abs {V} \\
			\ell_1 + \ell_2 - \abs S = \ell
		}} \tuplike{
			\Gamma_j(\ell_1, S) + \Gamma_k(\ell_2, S)
			- \alpha(\partial S \cap (S \times (X_i \setminus S))) 
		}
	\end{equation}
	We omit the details on how the necessary values of $\alpha$ are computed in their case;
	it suffices to see that a rough analysis of the equation above, assuming that we are given
	the value of the $\alpha$ expression, is $\O(n^2)$ per entry for a single join node.
	It is also easy to see that there are indeed instances where the computation of an entry
	for $(\ell, S)$ takes $\Theta(n^2)$ time. This yields an overall running time
	of $\O(2^tn^3)$ as stated by \textcite{DBLP:journals/siamcomp/JansenKLS05}.

	\textcite{DBLP:journals/tcs/HanakaKS21} provided an improved analysis for the algorithm of \textcite{DBLP:journals/siamcomp/JansenKLS05}.
	They defined $\nu_i$ to be the sum of all $\abs{X_j}$ for nodes $j$ that 
	are below $i$. It is not hard to see that the running time for the computation
	of a single entry of a join node $i$ with left child $j$ and right child $k$
	can be done in time $\O(\nu_j\nu_k)$. Using a labeling argument they then proved
	that
	\[ \sum\limits_{i ~:~ \text{join node with children $j$,$k$}}\nu_j\nu_k \leq (nt)^2. \]
	Their idea is that after labeling the vertices in all bags (possibly giving the same vertex 
	different labels for different bags), every pair of those labels can occur at at most
	one join node. The consequence of the above statement is that the worst case for join nodes
	cannot occur too often; overall, all entries of all join nodes can be computed in time
	$\O(2^t(nt)^2)$.

	Our approach is now to reformulate the recurrence by something that can be thought of
	as an index shift; for each node $i$ we define a table
	\let\Gamma=\oldGamma
	\begin{align} \label{eq:Shift}
		& \Gamma_i \colon \sing{0, \dots, \abs {F_i}} \times 2^{X_i} \to \Q
		\notag
		\\
		&\Gamma_i(\ell, S) \defequiv \max_{
			\hat S \in \binom {F_i} \ell 
		} \alpha\enpar{\partial \enpar{S \uplus \hat S} \cap Y_i^2}
	.\end{align}
	In comparsion to \eqref{eq}, there are two differences.
	\begin{enumerate}
		\item The indices have a different meaning: In $\Gamma_i(\ell, S)$ we store
		the value of the best cut (with respect to $f$) that consists of the set
		$S \subseteq X_i$ and $\ell$ further vertices that
		\emph{occur in bags below $i$, but not in $X_i$}.
		\item The table's size is now $\O(\abs {F_i} 2^t)$; this is not only
		smaller but also for every entry $(\ell, S)$ there is a cut consisting
		of $\ell$ vertices from $F_i$ and the vertices from $S$ (hence we do not
		have to consider those special cases of undefinedness as it had to be done
		in \cite{DBLP:journals/siamcomp/JansenKLS05}).
	\end{enumerate}
	For the modified recursion, the join nodes are still the bottleneck; their recurrence is\footnote{Since we only did an index shift, we can reuse the recurrence from \cite{DBLP:journals/siamcomp/JansenKLS05} by applying the same shift to it}
	\[
		\Gamma_i(\ell, S) \equiv \bigg(\max_{\substack{
			0 \leq \ell_1 \leq \abs{F_j} \\
			0 \leq \ell_2 \leq \abs{F_k} \\
			\ell_2 + \ell_1 = \ell
		}}
		\tuplike{
			\Gamma_j(\ell_1, S)
		+
			\Gamma_k(\ell_2, S)
		}
		\bigg) - \alpha(\partial S \cap Y_i^2).
	\]
	The key observation is that the running time is dominated by computing the $\max$ expression\footnote{
		Note that $\partial S$ can only take on $\O(2^t)$ different values at some fixed node
		$i$. We use a simple DP to precompute the $\alpha(\cdot)$ terms efficiently (for a fixed node in time $\O(2^tn)$).
	}, which depends linearly on $\abs{F_j} \cdot \abs {F_k} = \abs{F_j \times F_k}$. We can now show that
	all of those occurring Cartesian products are disjunct:
	\begin{proposition} \label{stmt:pairs}
		For each pair $(v_1, v_2) \in V^2$ there is at most one join node
		$i$ with left child $j$ and right child $k$ such that $(v_1, v_2) \in
		F_j \times F_k$. \stopp
	\end{proposition}
	\begin{proof}
		Proof by contradiction. Assume there was a join node $i' \neq i$ with left child
		$j'$ and right child $k'$ such that $(v_1, v_2) \in F_{j'} \times F_{k'}$.
		Then, by \hyperref[defi:tree-decomposition:coherence]{Coherence},
		either $i'$ is below $i$ or $i$ is below $i'$. We assume
		without loss of generality that $i'$ is below $i$. Moreover, we assume
		without loss of generality that $i'$ is somewhere in the \emph{left}
		subtree of $i$. As $(v_1, v_2) \in F_{j} \times F_{k}$ by assumption,
		we have in particular $v_2 \in F_k$, and, additionally taking into
		account that $F_j \cap F_k = \emptyset$ by \hyperref[defi:tree-decomposition:coherence]{Coherence}, $v_1 \notin F_k$.
		Since $i'$ is in the left subtree of $i$, we also have $F_{j'} \subseteq F_{i'}
		\subseteq F_j$, hence $v_1 \notin F_{j'}$. This is a contradiction to
		$(v_1, v_2) \in F_{j'} \times F_{k'}$.
	\end{proof}\noindent
	Using this statement we can now deduce that 
	\begin{equation} \label{eq:nn}
		\sum\limits_{i~:~\text{join node with children $j$,$k$}} \hskip -1em \abs{F_j \times F_k}
		= \abs*{\biguplus\limits_{i~:~\text{join node with children $j$,$k$}} \hskip -1em \big(F_j \times F_k)}
		\leq \abs{V^2} = n^2.
	\end{equation}

	We can thus deduce the overall running time of computing all entries for all
	join nodes is $\O(2^t n^2)$, as we need time $\O(2^t \abs{F_j \times F_k})$
	for a single join node.

	As the running time for the other node types obviously remain unchanged, this yields an algorithm with overall running time $\O(2^tn^2)$
	for Max-Bisection.

	\section{Our Framework} \label{sec:framework}
	\let\Gamma=\oldGamma
	In this section we discuss how we can generalize the idea from the previous section to other cut problems.
	More precisely, we present a framework that can solve several cut-problems
	(for directed, arbitrarily-weighted graphs $G = (V,E)$ with weight function $w : E \to \Q$) 
	in time $\O(2^tn^2)$ if a small tree decomposition of width $t$ is given
	as part of the input.
	Without loss of generality we assume that this small tree decomposition is also a nice
	tree decomposition (if not, we could simply use \autoref{stmt:nice} to convert it 
	accordingly in sufficiently small time).

	The main obstacle is finding an abstraction of the algorithm for Max-Bisection
	that maintains the running time, but also allows us tackle all the listed problems.
	Especially extracting the formal arguments hidden implicitly in existing algorithms
	turned out to be a non-trivial task.

	For our framework, we assume that we are given an objective function
	$f \colon \N_0 \times \Q \to \bar \Q$ that is either monotonic or antitonic\footnote{$x \leq y \implies f(a,x) \geq f(a,y)$}
	in its second argument,
	and a validator function $\valid \colon \N_0 \to \sing{\true, \false}$.
	We use a 
	function of arity 2 to be able to not only model Max-Bisection and similar, but also
	e.\@{}g.\@{} Sparsest-Cut, where the objective depends on the size of the cut
	\emph{and} the number of vertices selected. The validator function is needed 
	e.\@{}g.\@{} for Max-Bisection, as we somehow have to tell the framework which
	entries correspond to feasible solutions and which are infeasible; for Max-Bisection
	we would set $\valid(x) \defeq (\abs{x - (n - x)} \leq 1)$.
We assume that both, $f$ and $\valid$, can be evaluated in time $\O(1)$.

	Our task is now to compute an element of all \emph{possible preimages} 
	(in the sense: there exists a corresponding cut) 
	$\Set{ \tup{ \abs S }{ \alpha(\partial S) } | S \subseteq V,\,\valid(\abs S)}$
	of the objective function $f$ that maximizes $f$.
	We can reformulate this task in a more elegant way by introducing the total
	quasiorder
	$(\sqleq)\subseteq (\sing{0, \dots, n} \times \Q) \uplus \sing \bot)^2$ 
	defined by
	\[a \sqleq b \iff a = \bot \lor \big(
			a \neq \bot \neq b \land f(a) \leq f(b)
		\big). \]
	The intuition behind that quasiorder is as follows: if we compare two values $a,b \neq \bot$, then $a \sqleq b$ iff $f(a) \leq f(b)$, that is,
	we compare (non-$\bot$) values by their image under $f$.

	The idea of the new symbol $\bot$ is to represent the case where there is no
	feasible solution and hence no possible preimage to $f$ as we have to deal
	with that case, too.

	Our task is now to compute (where $\bigsqcup$ is the supremum operator regarding $\sqleq$).
	\begin{equation} \label{eq:phi1}
		\Phi \equiv
			\bigsqcup_{\substack{
				S \subseteq V \\
				\valid(\abs S)
		}}
		\tup{
			\abs S
		}{
			\alpha(\partial S)
		}.
	\end{equation}
By definition, $\bot$ is the smallest element of our order,
	so if there is a feasible solution, the result cannot be $\bot$.

	From a strict mathematical perspective, \autoref{eq:phi1} is incorrect as
	in general there is no such thing as a unique supremum for a total quasiorder
	(there might be multiple possible preimages of $f$ taking on the optimal value).
	Taking this into account would make the description of our approach way more complicated,
	as we would need to reason about equivalence classes and eventually give a
	recurrence to compute a representant of the class of element optimizing the objective function. Thus, we identify elements and their corresponding equivalence class (set of possible preimages that have the same objective value) in this paper.

	For the sake of shortness, for a subset of edges $M \subseteq E$ we write $\alpha_i(M) \defeq \alpha(M \cap Y_i^2)$.
	We set up the dynamic program similar to the one for Max-Bisection, that is for a node $i$ we define
	\begin{align} \label{eq:Gamma}
		& \Gamma_i \colon \sing{0, \dots, \abs {F_i}} \times 2^{X_i} \to \sing{0, \dots, n} \times \Q
		\notag
		\\
		&\Gamma_i(\ell, S) \defequiv \bigsqcup_{
			\hat S \in \binom {F_i} \ell 
		} \tup{\ell + \abs S}{\alpha_i \enpar{\partial \enpar{S \uplus \hat S}}}
	.\end{align}
	In comparison to \autoref{eq:Shift} there are two differences.
	\begin{enumerate}
		\item An entry is no longer just the size of the corresponding cut,
		but a $2$-tuple consisting of the number of vertices selected and
		the size of the cut.
		\item Instead of storing values for the largest cut, as we did for Max-Bisection,
		we store the tuple that	maximizes the function $f$.
	\end{enumerate}
We can now rewrite \autoref{eq:phi1} in terms of $\Gamma$ (recall that $r$
	is the root of the given nice tree decomposition):
	\begin{lemma}
		\label{stmt:rewrite}
		\begin{equation}
			\Phi \equiv \bigsqcup_{\substack{
			S \subseteq X_r \\
			0 \leq \ell \leq \abs {F_r} \\
			\valid(\ell + \abs S)
		}}
		\Gamma_r(\ell, S) 
		\end{equation}
	\end{lemma}
	\submorarxiv{\omitted}
	{\ifllncs{\begin{proof}}{\begin{proof}[\temptext]}
	\begin{multline*}
		\bigsqcup_{\substack{
			S \subseteq X_r \\
			0 \leq \ell \leq \abs {F_r} \\
			\valid(\ell + \abs S)
		}} \B_r(\ell, S)
	=
		\bigsqcup_{\substack{
			S \subseteq X_r \\
			0 \leq \ell \leq \abs {F_r} \\
			\valid(\ell + \abs S)
		}} \bigsqcup_{
		\hat S \in \binom {F_r} \ell 
	} \tup{\ell + \abs S}{\alpha_r (\partial (S \uplus \hat S)}
	=
	\bigsqcup_{\substack{
			S \subseteq X_r \\
			0 \leq \ell \leq \abs {F_r} \\
			\valid(\ell + \abs S) \\
		\hat S \in \binom {F_r} \ell 
	}} \tup{\ell + \abs S}{\alpha_r (\partial (S \uplus \hat S)}
	\\ =
	\bigsqcup_{\substack{
		S \subseteq X_r \\ \hat S \subseteq F_r
	}} \tup{\abs S + \abs {\hat S}} {\alpha_r(\partial(S \uplus \hat S))}
	= \bigsqcup_{\substack{S \subseteq V \\ \valid(\abs S)}}
	\tup {\abs S} {\alpha_r(\partial S)}
	= \Phi \qedhere \end{multline*}
\end{proof}%
 }

	\noindent It is easy to see that, if we are given the values
	$\Gamma_r(\ell, S)$ for all $0 \leq \ell \leq \abs {F_r}$ and all $S \subseteq X_r$,
	we can compute $\Phi$ in time $\O(2^tn)$. We claim that we can compute the table
	$\Gamma$ for all nodes together in overall time $\O(2^tn^2)$, implying that $\Phi$
	can be computed in time $\O(2^tn^2)$. To see this, we now show that we can
	use a dynamic program to compute the table $\Gamma$ and eventually $\Gamma_r$ in the
	desired time. Therefor, we first set up recurrences for each node type that
	we can use to efficiently compute the value for a node $i$ of this type, given that
	we already know all the values \emph{below} $i$.

	For our approach there is an important property of $\sqleq$: Basically, we are
	able to move the addition with a constant tuple outside of the supremum operator,
	if the elements all have the same first component.
	\begin{lemma}
		\label{stmt:compatible}
		For any $a$, any finite set $M \subseteq \sing a \times \Q$ and any $z$ it holds that
		\[ z \oplus \bigsqcup_{x \in M} x \equiv \bigsqcup_{x \in M} \tuplike {z \oplus x} \stopphier \]
	\end{lemma}
	\begin{proof}
		As $M$ is finite, it suffices to show\footnote{For the reader not familiar with order theory: The binary supremum
		is an associative and commutative map. If for every pair of elements there is a supremum,
		that is, a smallest element that is larger than both elements of the pair,
		then so it does for any finite set $M$. This can be shown by a simple
		inductive argument using the aforementioned associativity/commutativity.}
		this property for the binary
		supremum $\sqcup$. Let $z = (b,w)$ and $(a,x) \in M$, $(a,y) \in M$.
		If $x = y$ or $f(a, x) = f(a, y)$, the property is trivial. Thus we may
		assume without loss of generality that $x > y$ and $f(a, x) \neq f(a, y)$.
		We now have two cases, depending on $h \mapsto f(a, h)$.
The first case
		is that $h \mapsto f(a, h)$ monotonic. Then our assumption implies that
		$f(a, x) > f(a, y)$ and hence:
		\[
			(b,z) \oplus \enpar{(a,x) \sqcup (a,y) }
			= (b,z) \oplus (a,x)
			= (a + b, x + z)
			= (a + b, x + z) \sqcup (a + b, y + z) 
		\]
		The last step follows from the monotony, as $x + z \geq y + z$.
		The second case, where $h \mapsto f(a, h)$ is antitonic, can be proven
		analogously.
	\end{proof}\noindent
	This property is absolutely crucial as it gives us some freedom for transformations of
	\autoref{eq:Gamma}; it is applied multiple times in the proofs in the appendix.
	Intuitively, this lemma tells us that the optimization process works no different
	than it does e.\@{}g.\@{} Max-Cut or Max-Bisection; if we fix the number of vertices
	we choose, given a set of cuts to pick from the cut maximizing $f$ is either
	the cut of largest size or the cut of smallest size.

	\ifsubm{\subsection*{Recurrences}}{\subsection{Recurrences}}
	We now set up a recurrence, depending on the node type, to compute the table $\Gamma$
	by traversing the tree decomposition in a bottom-up fashion. In this section
	we give the intuition behind the recurrences step by step. For correctness
	proofs of the equations (which are very technical) we refer to the
	\ifsubm{\fullversion\ of the paper}{appendix}.

	\paragraph*{Leaf Node} Let $i$ be a leaf node. Then $F_i = \emptyset$.
	Thus, $\Gamma_i$ is only defined for $\ell = 0$ and $S \subseteq X_i$.
	If $S = \emptyset$, there are no edges in the cut, hence $\Gamma_i(0,\emptyset) \equiv (0,0)$.
	If $S = \sing v$ consists of a single node, we can simply go through all its
	edges, that is set
	\begin{equation} \label{eq:leaf:1}
		\Gamma_i(0, \sing v) \equiv \Big(1, \sum_{\substack{v' \in X_i \\ vv' \in E}} w(vv') \Big)
	.\end{equation}
	Now let $S = S' \uplus \sing v$ where $S' \neq \emptyset$. We will now argue
	how we can compute the value $\Gamma_i(0, S)$ given the value $\Gamma_i(0, S')$.
	In the situation considered for $\Gamma_i(0, S')$ we have $v \notin S'$. 
	If we move $v$ into the selection and are able to track and compute the changes,
	we can also compute $\Gamma_i(0, S)$. After moving $v$ into the selection, there
	might be edges from $S'$ to $\sing v$ (which all have been in the cut before); those edges
	are no more in the cut for $S$. Also, there might also be some new edges in the
	cut, more precisely all edges from $v$ to $X_i \setminus S$. There are no more
	new edges in the cut and no other edges are removed from the cut. This yields
	\begin{equation} \label{eq:leaf:2}
		\Gamma_i(0, \underbrace{\overann{S'}{\neq \emptyset} \uplus \sing v}_S) \equiv \Gamma_i(0, S')
			\oplus \tup{
				\overann{1}{\abs S - \abs {S'} = 1,\text{~we now add the new vertex $v$}}
			}{
				\underbrace{\sum_{\substack{v' \in X_i \setminus S' \\ vv' \in E}} w(vv')}
				_{\text{edges from $v$ to $X_i \setminus S'$}}
				- \underbrace{\sum_{\substack{v' \in S' \\ v'v \in E}} w(v'v)}
				_{\text{edges from $S'$ to $v$}}
			}
	.\end{equation}
	\paragraph*{Forget Node} Let $i$ be a forget node with child $j$. Then there
	is $v \in V$ such that $X_i \uplus \sing v = X_j$. Now, for the computation of the
	entries $\Gamma_i(\ell, S)$ we only have to deal with one question: is it better
	to include $v$ into the selection or not? As $v \in F_i$, this question
	is only relevant if $\ell \geq 1$; if $\ell = 0$,
	then $\Gamma_i(\ell, S) \equiv \Gamma_j(\ell, S)$.
	Now let $\ell \geq 1$. If we included $v$ into our selection, we have to
	include $\ell - 1$ further vertices from $F_j$; if we do not include $v$, we
	have to include $\ell$ further from $F_j$. We now simply pick the better result
	of both options. Overall, this yields the following recurrence:
	\begin{equation} \label{eq:forget}
		\Gamma_i(\ell, S) \equiv \begin{cases}
			\Gamma_j(\ell-1, S \uplus \sing v) \sqcup \Gamma_j(\ell, S) & \ell \geq 1 \\
			\Gamma_j(\ell, S) & \ell = 0
		\end{cases}.
	\end{equation}
	\paragraph*{Introduce Node} Let $i$ be an introduce node with child $j$. Then
	there is $v \in V$ such that $X_i = X_j \uplus \sing v$. When computing
	$\Gamma_i(\ell, S)$ we have to make a case distinction whether $v$ is in $S$
	or not.

	In the first case $v \in S$ we rely on $\Gamma_j(\ell, S \setminus \sing v)$
	to compute the value. As we consider a selection of $S$ and $\ell$ further
	vertices of $F_i$ as opposed to the entry $\Gamma_j(\ell, S \setminus \sing v)$
	which only considers a selection of $S \setminus \sing v$ and $\ell$ further
	vertices of $F_j = F_i$, we have to add $1$ to the first component of
	$\Gamma_j(\ell, S \setminus \sing v)$ to account for the additional vertex $v$.
	For the second component, the size of the cut, we only have to add the
	weight of all edges from $v$ to $X_j \setminus S$. Note that there are no
	weights that we need to subtract; while there might be edges between
	$v$ and $S$, those are not considered in the computation of
	$\Gamma_j(\ell, S \setminus \sing v)$
	as the vertex $v$ does not appear in $Y_j$ (due to \hyperref[defi:tree-decomposition:coherence]{Coherence}
	and \hyperref[defi:tree-decomposition:edge_coverage]{Edge Coverage}).
	Thus, for $v \in S$ we get
	\begin{equation} \label{eq:introduce:1}
		\Gamma_i(\ell, S) \equiv
			\tup{ \overann{1}{\text{\qquad\qquad\qquad\quad we now add the new vertex $v$ to the selection $S \setminus \sing v$}} }{
			\underbrace{\alpha_i(\partials v \cap (\sing v \times (X_i \setminus S))}
			_{\mathclap{\text{edges from $v$ to $X_i \setminus S$; we will later discuss how
			to compute this value}}}
			} \oplus
			\Gamma_j(\ell, S \smin v)
	.\end{equation}

	Now let us consider the case $v \notin S$. The argumentation is very similar.
	As $v \notin S$, there is no new vertex we add to the selection. However,
	there are possibly new edges and that change the value of the cut: the edges
	from $S$ to $\sing v$. This gives us the following recurrence for the case
	$v \notin S$:
	\begin{equation} \label{eq:introduce:2}
	\Gamma_i(\ell ,S) \equiv 
		\tup  0 { \underbrace{
		\alpha_i(\partial S \cap (X_i\times \sing v))}
		_{\mathclap{\text{edges from $S$ to $\sing v$; we will later discuss how
			to compute this value}}}
		} \oplus \Gamma_j(\ell, S)
	.\end{equation}

	In both cases there is a (case-dependent) additive term that has to b eadded to
	an entry of $\Gamma_j$. Observe that this additive term is independent of
	$\ell$.
As it turns out, this additive term (depending only on $S$) can be rewritten as
	\begin{equation} \label{eq:introduce:3}
		\pi_2 \big(\Gamma_i(\ell, S) \ominus \Gamma_j(\ell, S) \big)
	\end{equation}
	for any arbitrary, but fixed $\ell$;
	for details we refer to
	\ifsubm{the \fullversion}{\autoref{stmt:eq:introduce:3} in the appendix}.
	Thus, we can do the following: for $\ell = 0$ we need to explicitly
	compute the additive term for all $S \subseteq X_i$.
	For all $\ell \geq 1$ we can simply use 
	$\pi_2(\Gamma_i(\ell-1, S) \ominus \Gamma_i(\ell-1, S))$ for all $S \subseteq X_i$.

	\paras{Join Node} Let $i$ be a join node with left child
	$j$ and right child $k$. Then $X_i = X_j = X_k$. The value 
	$\Gamma_i(\ell, S)$ can then be interpreted as the tuple maximizing $f$
	over all selections $\hat S \subseteq Y_i$ where there are $\ell_1$ vertices
	from $F_j$ and $\ell_2$ from $F_k$ for every
	$0 \leq \ell_1 \leq \abs {F_j}$,
	$0 \leq \ell_2 \leq \abs {F_k}$ where
	$\ell_1 + \ell_1 = \ell$. This is closely related to
	\begin{equation*}
		\Gamma_i(\ell, S) \equiv \bigsqcup_{\substack{
			0 \leq \ell_1 \leq \abs{F_j} \\
			0 \leq \ell_2 \leq \abs{F_k} \\
			\ell_2 + \ell_1 = \ell
		}}
		\tuplike{
			\Gamma_j(\ell_1, S)
		\oplus
			\Gamma_k(\ell_2, S)
		}
	;\end{equation*}
	it is not hard to see that this equation almost computes the desired
	tuple, with the exception that each vertex in $S$ and each edge of $E \cap X_i^2$
	that is in the cut is counted twice. Thus, by subtracting the doubly counted
	vertices/edge weights and applying \autoref{stmt:compatible} we get
	\begin{equation} \label{eq:join:1}
		\Gamma_i(\ell, S) \equiv \bigg(\bigsqcup_{\substack{
			0 \leq \ell_1 \leq \abs{F_j} \\
			0 \leq \ell_2 \leq \abs{F_k} \\
			\ell_2 + \ell_1 = \ell
		}}
		\tuplike{
			\Gamma_j(\ell_1, S)
		\oplus
			\Gamma_k(\ell_2, S)
		}
		\bigg) \ominus \tup{ \abs S }{ \alpha_i(\partial S) }
	.\end{equation}
	
	\ifsubm{The running time analysis is analogous to
		the analysis of the algorithm for Max-Bisection we
		presented before. We thus skip this part here (due to space restrictions); if interested, the reader
		can find a detailed analysis of every case in the
		appendix}{\subsection{Running Time}
		We claim that our algorithm has an overall running time of
	$\O(2^tn^2)$. To prove this we first need to deal with some preprocessing
	steps. The first step was already mentioned: converting the tree decomposition
	into a nice tree decomposition. This can be done in $\O(nt^2)$ according to
	\autoref{stmt:nice}. We model sets as bit vectors of length $n$; we omit the
	(very technical) details on how set operations can be implemented on bit vectors
	such that they all take only constant time.
	The second preprocessing step is creating an adjacency matrix of the graph
	in time $\O(n^2)$; this is necessary to be able to compute the explicit sums
	occurring in our recurrences in time $\O(n)$ each. We now claim that our dynamic
	program takes the running time specified in \autoref{figure:running_time_goals}
	per node, depending on the node's type. As we only have $\O(n)$ nodes in the
	decomposition (it is a small decomposition), this directly implies that
	the computation of $\Gamma_i$ only takes time $\O(2^tn^2)$.

	\begin{figure}[t]
		\centering
		\begin{tabulary}{\textwidth}{|L||LL|}
			\hline
			Leaf node & $\O(2^t n)$ & per node 
			\\ \hline
			Forget node & $\O(2^t \abs{F_i})$ & per node
			\\ \hline
			Introduce node & $\O(2^t n)$ & per node
			\\ \hline
			\multirow{2}{*}{Join node} & $\O(2^t \abs{F_j \times F_k})$ & per node
			\\ 
			& $\O(2^t n^2)$ & for all nodes
			\\ \hline
		\end{tabulary}
		\caption{Running times of the computation of $\Gamma_i$ for
		a node $i \in I$ (if there is one child $i$, let $j$ be that child, and if
		there are two children, let $j,k$ be those children of $i$) depending
		on the node's type}
		\label{figure:running_time_goals}
	\end{figure}

	\paras{Leaf Node} Let $i$ be a leaf node. Then $F_i = \emptyset$.
	Hence $\Gamma_i$ has $\O(2^t)$ entries. We make use of \multiautoref{eq:leaf:1, eq:leaf:2}
	to compute these entries; both equations have a constant number of sums, and
	each sum can be computed in $\O(n)$ using the adjacency matrix and basic
	set operations.

	\paras{Forget Node}
	Let $i$ be a forget node with child $j$. Then $X_i \uplus \sing v = X_j$
	for some $v \in V$. There are $\O(2^t \abs{F_i})$ entries in $\Gamma_i$.
	As we only use \autoref{eq:forget} to compute the values, which takes
	constant time per value, the overall running time for a single forget node
	is $\O(2^t \abs{F_i})$.
	
	\paras{Introduce Node}
	Let $i$ be an introduce node with child $j$. Then $X_i = X_j \uplus \sing v$
	for some $v \in V$. There are $\O(2^t \abs{F_i}) \subseteq \O(2^tn)$ entries
	in $\Gamma_i$ we need to compute. For the entries $\Gamma_i(0, S)$ for
	$S \subseteq X_i$ we use \multiautoref{eq:introduce:1, eq:introduce:2} and explicitly compute the
	$\alpha_i$ expressions as sums in $\O(n)$ each. There are at most
	$\O(2^t)$ entries for $\ell = 0$, thus for all these entries we need at most
	$\O(2^tn)$ time.
	For all other $\O(2^tn - 2^t) = \O(2^tn)$ entries we make use
	of \autoref{eq:introduce:3} which allows us now to use 
	\multiautoref{eq:introduce:1, eq:introduce:2} to compute the remaining entries
	in time $\O(1)$ each. The overall running time for a single introduce node
	is hence $\O(2^tn)$.

	\paras{Join Node}
	Let $i$ be a join node with left child $j$ and right child $k$.
	Then $X_i = X_j = X_k$. There are $\O(2^t \abs{F_i}) \subseteq \O(2^t \abs{F_j}\cdot\abs{F_k})$
	entries we need
	to compute. We only use \autoref{eq:join:1} to compute them. For a fixed $S \subseteq
	X_i$, we need to compare $\abs{F_j} \cdot \abs{F_k}$ different pairs in that equation.
	This means
	that the running time -- per entry -- is in $\O(\abs{F_j} \cdot \abs{F_k})
	= \O(\abs{F_j \times F_k})$. Using \autoref{eq:nn} (which is a consequence of
	\autoref{stmt:pairs}) it follows that the overall running time for all entries
	of all join nodes is bounded by $\O(2^tn^2)$. }
\exceptsubm{\section{Applications} \label{sec:instantiations}
	We now give applications to several cut problems. Recall that our framework
	computes $\Phi$, a 2-tuple that maximizes a given map $f$. At this point we
	want to point out that $\Phi = \bot$ can only occur if there is no partition
	of the vertices whose cardinality satisfies $\valid$. This does not occur
	in our applications.\footnote{In general, we cannot think of any reason why
	this case should occur in practice; it seems to us that if it occurs this
	is likely to indicate an error of the model used.}
	Note that for any instantiation of the framework we have
\begin{equation}
	f(\Phi)
		= f\bigg(\bigsqcup_{\substack{S \subseteq V \\ \valid(\abs V)}}
			\tup {\abs S} {\alpha(\partial S)}
		\bigg)
		= \max_{\substack{S \subseteq V \\ \valid(\abs S)}}
		f \tup{\abs S}{\alpha(\partial S)}.
\end{equation} 
	\para{Max-Cut}
	Set $f (\nu, s)\defeq  s$ and $\valid(x) \defeq \true$.
	It is not hard to see that $f$ is monotonic in the second argument and hence
	fulfils the requirements of our framework.
\begin{lemma}
	\[ f(\Phi) = \max_{S \subseteq V} \alpha(\partial S) \stopphier \]
\end{lemma}
\begin{proof}
	\[
		f(\Phi)
		= \max_{\substack{S \subseteq V \\ \valid(\abs S)}}
			f \tup{ \abs S }{ \alpha(\partial S) }
		= \max_{S \subseteq V} f \tup{ \abs S }{ \alpha(\partial S) }
		= \max_{S \subseteq V} \alpha(\partial S) \qedhere
	\]		
\end{proof} 
	\para{Max-Bisection'}
	Set $f(\nu, s) \defeq s$ and $\valid(x) \defeq
	\big( \abs*{2x - n} \leq 1 \big)$. It is not hard to see that $f$ is 
	monotonic in the second argument and hence fulfils the requirements of our
	framework.
\begin{lemma}
	\[ f(\Phi) = \max_{\substack{S \subseteq V \\ \abs*{\abs{S} - \abs{V \setminus S}} \leq 1}} \alpha(\partial S) \stopphier \]
\end{lemma}
\begin{proof}
	\[
		f(\Phi)
		 = \max_{\substack{S \subseteq V \\ \valid(\abs S)}}
		 f \tup{ \abs S }{ \alpha(\partial S) }
		 = \max_{\substack{S \subseteq V \\ \abs*{2\abs{S} - n} \leq 1}} \alpha(\partial S)
		 = \max_{\substack{S \subseteq V \\ \abs*{\abs{S} - \abs{V \setminus S}} \leq 1}} \alpha(\partial S) \qedhere
	\]
\end{proof} 
	\para{$\beta$-Balanced-Min-Cut'}
	Set $f (\nu, s) \defeq -s$ and $\valid(x) \defeq \big(\beta n \leq x \leq (1-\beta)n \big)$; the negation is required to model a minimization problem in our maximization framework. It is not hard to see that $f$ is 
	antitonic in the second argument and hence fulfils the requirements of our
	framework.
\begin{lemma}
	\[ - f(\Phi) = \min_{\substack{S \subseteq V \\ \beta \cdot \abs V \leq \abs S \leq (1 - \beta) \cdot \abs V}} \alpha(\partial S) \stopphier \]
\end{lemma}
\begin{proof}
	\[
		- f(\Phi)
		 = - \max_{\substack{S \subseteq V \\ \valid(\abs S)}}
		 	f \tup{ \abs S }{ \alpha(\partial S) }
		 = \min_{\substack{S \subseteq V \\ \beta n \leq \abs S \leq (1-\beta) n}}
		 	- f \tup{ \abs S }{ \alpha(\partial S) }
		 = \min_{\substack{S \subseteq V \\ \beta \abs V \leq \abs S \leq (1-\beta) \abs V}}
		 \alpha(\partial S) \qedhere
	\]
\end{proof} 
	\para{Min-Edge-Expansion'}
	Set $f (\nu, s) \defeq \begin{cases}
		-\infty & \nu = 0 \\
		-\frac{s}{\nu}& \totherwise
	\end{cases}$.
	We claim that $f$ is antitonic in its
	second argument. To see this, fix $\nu$ such that $0 \leq \nu \leq n$.
	and let $g(s) \defeq f(\nu, s)$.
	If $\nu = 0$, then for all $s$ we have that  $g(s) = f(\nu, s) = -\infty$ is a constant
	map and hence antitonic. If $\nu > 0$, then then for all $s$ we have that $g(s) =  f(\nu, s)
	= -\frac s \nu$ is a linear map with negative sign and hence
	also antitonic.

	We now need the validator function to exclude all sets $S$ where
	$\abs S > \abs{V \setminus S}$; we do not have to explicitly exclude the
	empty set as our measure function evaluates to $- \infty$ for those cases
	(but we may do so). Thus, we might define $\valid(x) \defeq \big(x \leq n - x\big)$. Then the value of the Min-Edge-Expansion can be
	recovered from the result of our framework $\Phi$ by computing
	$-f(\Phi)$. That value cannot be $\infty$ unless the graph is empty.
\begin{lemma}
	If $n \geq 1$, then $ \displaystyle -f(\Phi)
	= \min_{\substack{\emptyset \neq S \subseteq V \\ \abs S \leq \abs{V \setminus S}}}
	\frac{\alpha(\partial S) }{\abs S} $. \stopp
\end{lemma}
\begin{proof} \begin{multline*}
		-f(\Phi)
		= -\max_{\substack{S \subseteq V \\ \valid(\abs S)}}
			f \tup{ \abs S }{ \alpha(\partial S) }
		= \min_{\substack{S \subseteq V \\ \abs S \leq n - \abs S}}
			-f \tup{ \abs S }{ \alpha(\partial S) }
		\\ = \min_{\substack{S \subseteq V \\ \abs S \leq \abs {V \setminus S}}} \begin{cases}
				\infty & \abs S = 0 \\
				\frac{\alpha(\partial S)}{\abs S} & \totherwise
			\end{cases}
		= \min \Sing{ \infty, \min_{\substack{\emptyset \neq S \subseteq V \\ \abs S \leq \abs {V \setminus S}}} \frac{\alpha(\partial S)}{\abs S}}
		\overset{n \geq 1}=
		\min_{\substack{\emptyset \neq S \subseteq V \\ \abs S \leq \abs {V \setminus S}}} \frac{\alpha(\partial S)}{\abs S}
	\myqedhere\end{multline*}
\end{proof} 
	\para{Sparsest-Cut'}
	Sparsest Cut' is quite similar to Min-Edge-Expansion'. We set
	$f(\nu, s) \defeq \begin{cases}
		-\infty & \nu \in \sing{0,n} \\
		-\frac{s}{\nu (n - \nu)}& \totherwise
	\end{cases}$. This is again an antitonic map, following the same
	argumentation as for Min-Edge-Expansion'. The validator function becomes
	trivial for Sparsest-Cut', that is $\valid(x) \defeq \true$.
	Recovering the Sparsest-Cut' from $\Phi$ works the same way as before:
	it can be expressed as $-f(\Phi)$ which, again, cannot be $\infty$
	unless the graph has at most one vertex.
	\begin{lemma}
	If $n \geq 2$, then $ \displaystyle -f(\Phi)
	= \min_{\emptyset \neq S \subsetneq V}
	\frac{\alpha(\partial S) }{\abs S \cdot \abs {V \setminus S}} $. \stopp
\end{lemma}
\begin{proof}\begin{multline*}
		-f(\Phi) =
		-\max_{\substack{S \subseteq V \\ \valid(\abs S)}}
			f\tup{ \abs S }{ \alpha(\partial S) }
		= \min_{S \subseteq V}
			-f\tup{ \abs S }{ \alpha(\partial S) }
		= \min_{S \subseteq V} \begin{cases}
				\infty & \abs S \in \sing{0, n} \\
				\frac{\alpha(\partial S)}{\abs S \cdot \abs{V \setminus S}} & \totherwise
			\end{cases}
		\\ = \min \Sing{ \infty, \min_{\emptyset \neq S \subseteq V} \frac{\alpha(\partial S)}{\abs S \cdot \abs{V \setminus S}}}
		\overset{n \geq 2}=
		\min_{\emptyset \neq S \subsetneq V} \frac{\alpha(\partial S)}{\abs S \cdot \abs{V \setminus S}}
	\myqedhere\end{multline*}
\end{proof} 
	\para{Densest-Cut'}
	As Densest-Cut' is the complementary problem of Sparsest-Cut', the instantiation
	of the framework is very similar: We set $f (\nu, s) \defeq \begin{cases}
		-\infty & \nu \in \sing{0,n} \\
		\frac{s}{\nu (n - \nu)}& \totherwise
	\end{cases}$ and $\valid(x) \defeq \true$. Then, $f(\Phi)$ is the
	value of the densest cut (and $-\infty$ if there is at most one vertex).
	\begin{lemma}
	If $n \geq 2$, then $ \displaystyle f(\Phi)
	= \max_{\emptyset \neq S \subsetneq V}
	\frac{\alpha(\partial S) }{\abs S \cdot \abs {V \setminus S}} $. \stopp
\end{lemma}
\begin{proof}~\vspace{-1.7em}
	\begin{multline*}
		f(\Phi) =
		\max_{\substack{S \subseteq V \\ \valid(\abs S)}}
			f\tup{ \abs S }{ \alpha(\partial S) }
		= \max_{S \subseteq V}
			f\tup {\abs S} {\alpha(\partial S) }
		= \max_{S \subseteq V} \begin{cases}
				-\infty & \abs S \in \sing{0, n} \\
				\frac{\alpha(\partial S)}{\abs S \cdot \abs{V \setminus S}} & \totherwise
			\end{cases}
		\\ = \max \Sing{ -\infty, \max_{\emptyset \neq S \subseteq V} \frac{\alpha(\partial S)}{\abs S \cdot \abs{V \setminus S}}}
		\overset{n \geq 2}=
		\max_{\emptyset \neq S \subsetneq V} \frac{\alpha(\partial S)}{\abs S \cdot \abs{V \setminus S}}
	\myqedhere \end{multline*}
\end{proof}
  }
	\section{Conclusion}
	We showed that -- given a small tree decomposition of width $t$ -- many cut problems
	can be solved in time $\O(2^tn^2)$ using our framework. To our knowledge, the running times achieved by our framework are better than the previously known algorithms for
	the considered problems. Moreover, this running
	time is unlikely to be improved significantly (improvements by factors $\poly t$ and/or $\poly \log n$ are not excluded) in general: An algorithm (solving Min-Bisection)
	that runs in time
	$\O(n^{2 - \eps} f(t))$ for some $f$ and some $\eps > 0$  would imply an algorithm
	of running time $\O(n^{2-\delta})$ for some $\delta > 0$ for
	$(\min, +)$-convolution~\cite{DBLP:journals/jcss/EibenLM21},
	which is considered unlikely~\cite{cygan_2017}. An algorithm (solving Max-Cut)
	cannot have a running time $\O((2-\eps)^t \poly n)$ for some $\eps > 0$ unless
	SETH fails~\cite{DBLP:journals/talg/LokshtanovMS18,DBLP:journals/tcs/HanakaKS21,DBLP:journals/jcss/ImpagliazzoP01,DBLP:conf/iwpec/CalabroIP09}. However, there might be problems
	that can be solved using our framework that we have not considered yet. Moreover,
	it might be possible to generalize the framework (with possibly worse running time)
	to e.\@{}g\@{} be able to also cover connectivity problems (see \cite{DBLP:journals/iandc/BodlaenderCKN15}). \onlysubm{\clearpage}
	\section{References}\printbibliography[heading=none]
\onlysubm{\end{document}\endinput}
\clearpage
	\onlylipics{\ifthenelse{\boolean{oddpage}}{}{\thispagestyle{plain}\phantom a\clearpage}}
	\appendix \part*{Appendix}\allowdisplaybreaks
	\addcontentsline{toc}{section}{Appendix}
	\onlysubm{
		\onlyllncs{\renewcommand\para[1]{\subsection{#1}}}
		\section{Applications} \label{sec:instantiations}
	We now give applications to several cut problems. Recall that our framework
	computes $\Phi$, a 2-tuple that maximizes a given map $f$. At this point we
	want to point out that $\Phi = \bot$ can only occur if there is no partition
	of the vertices whose cardinality satisfies $\valid$. This does not occur
	in our applications.\footnote{In general, we cannot think of any reason why
	this case should occur in practice; it seems to us that if it occurs this
	is likely to indicate an error of the model used.}
	Note that for any instantiation of the framework we have
\begin{equation}
	f(\Phi)
		= f\bigg(\bigsqcup_{\substack{S \subseteq V \\ \valid(\abs V)}}
			\tup {\abs S} {\alpha(\partial S)}
		\bigg)
		= \max_{\substack{S \subseteq V \\ \valid(\abs S)}}
		f \tup{\abs S}{\alpha(\partial S)}.
\end{equation} 
	\para{Max-Cut}
	Set $f (\nu, s)\defeq  s$ and $\valid(x) \defeq \true$.
	It is not hard to see that $f$ is monotonic in the second argument and hence
	fulfils the requirements of our framework.
\begin{lemma}
	\[ f(\Phi) = \max_{S \subseteq V} \alpha(\partial S) \stopphier \]
\end{lemma}
\begin{proof}
	\[
		f(\Phi)
		= \max_{\substack{S \subseteq V \\ \valid(\abs S)}}
			f \tup{ \abs S }{ \alpha(\partial S) }
		= \max_{S \subseteq V} f \tup{ \abs S }{ \alpha(\partial S) }
		= \max_{S \subseteq V} \alpha(\partial S) \qedhere
	\]		
\end{proof} 
	\para{Max-Bisection'}
	Set $f(\nu, s) \defeq s$ and $\valid(x) \defeq
	\big( \abs*{2x - n} \leq 1 \big)$. It is not hard to see that $f$ is 
	monotonic in the second argument and hence fulfils the requirements of our
	framework.
\begin{lemma}
	\[ f(\Phi) = \max_{\substack{S \subseteq V \\ \abs*{\abs{S} - \abs{V \setminus S}} \leq 1}} \alpha(\partial S) \stopphier \]
\end{lemma}
\begin{proof}
	\[
		f(\Phi)
		 = \max_{\substack{S \subseteq V \\ \valid(\abs S)}}
		 f \tup{ \abs S }{ \alpha(\partial S) }
		 = \max_{\substack{S \subseteq V \\ \abs*{2\abs{S} - n} \leq 1}} \alpha(\partial S)
		 = \max_{\substack{S \subseteq V \\ \abs*{\abs{S} - \abs{V \setminus S}} \leq 1}} \alpha(\partial S) \qedhere
	\]
\end{proof} 
	\para{$\beta$-Balanced-Min-Cut'}
	Set $f (\nu, s) \defeq -s$ and $\valid(x) \defeq \big(\beta n \leq x \leq (1-\beta)n \big)$; the negation is required to model a minimization problem in our maximization framework. It is not hard to see that $f$ is 
	antitonic in the second argument and hence fulfils the requirements of our
	framework.
\begin{lemma}
	\[ - f(\Phi) = \min_{\substack{S \subseteq V \\ \beta \cdot \abs V \leq \abs S \leq (1 - \beta) \cdot \abs V}} \alpha(\partial S) \stopphier \]
\end{lemma}
\begin{proof}
	\[
		- f(\Phi)
		 = - \max_{\substack{S \subseteq V \\ \valid(\abs S)}}
		 	f \tup{ \abs S }{ \alpha(\partial S) }
		 = \min_{\substack{S \subseteq V \\ \beta n \leq \abs S \leq (1-\beta) n}}
		 	- f \tup{ \abs S }{ \alpha(\partial S) }
		 = \min_{\substack{S \subseteq V \\ \beta \abs V \leq \abs S \leq (1-\beta) \abs V}}
		 \alpha(\partial S) \qedhere
	\]
\end{proof} 
	\para{Min-Edge-Expansion'}
	Set $f (\nu, s) \defeq \begin{cases}
		-\infty & \nu = 0 \\
		-\frac{s}{\nu}& \totherwise
	\end{cases}$.
	We claim that $f$ is antitonic in its
	second argument. To see this, fix $\nu$ such that $0 \leq \nu \leq n$.
	and let $g(s) \defeq f(\nu, s)$.
	If $\nu = 0$, then for all $s$ we have that  $g(s) = f(\nu, s) = -\infty$ is a constant
	map and hence antitonic. If $\nu > 0$, then then for all $s$ we have that $g(s) =  f(\nu, s)
	= -\frac s \nu$ is a linear map with negative sign and hence
	also antitonic.

	We now need the validator function to exclude all sets $S$ where
	$\abs S > \abs{V \setminus S}$; we do not have to explicitly exclude the
	empty set as our measure function evaluates to $- \infty$ for those cases
	(but we may do so). Thus, we might define $\valid(x) \defeq \big(x \leq n - x\big)$. Then the value of the Min-Edge-Expansion can be
	recovered from the result of our framework $\Phi$ by computing
	$-f(\Phi)$. That value cannot be $\infty$ unless the graph is empty.
\begin{lemma}
	If $n \geq 1$, then $ \displaystyle -f(\Phi)
	= \min_{\substack{\emptyset \neq S \subseteq V \\ \abs S \leq \abs{V \setminus S}}}
	\frac{\alpha(\partial S) }{\abs S} $. \stopp
\end{lemma}
\begin{proof} \begin{multline*}
		-f(\Phi)
		= -\max_{\substack{S \subseteq V \\ \valid(\abs S)}}
			f \tup{ \abs S }{ \alpha(\partial S) }
		= \min_{\substack{S \subseteq V \\ \abs S \leq n - \abs S}}
			-f \tup{ \abs S }{ \alpha(\partial S) }
		\\ = \min_{\substack{S \subseteq V \\ \abs S \leq \abs {V \setminus S}}} \begin{cases}
				\infty & \abs S = 0 \\
				\frac{\alpha(\partial S)}{\abs S} & \totherwise
			\end{cases}
		= \min \Sing{ \infty, \min_{\substack{\emptyset \neq S \subseteq V \\ \abs S \leq \abs {V \setminus S}}} \frac{\alpha(\partial S)}{\abs S}}
		\overset{n \geq 1}=
		\min_{\substack{\emptyset \neq S \subseteq V \\ \abs S \leq \abs {V \setminus S}}} \frac{\alpha(\partial S)}{\abs S}
	\myqedhere\end{multline*}
\end{proof} 
	\para{Sparsest-Cut'}
	Sparsest Cut' is quite similar to Min-Edge-Expansion'. We set
	$f(\nu, s) \defeq \begin{cases}
		-\infty & \nu \in \sing{0,n} \\
		-\frac{s}{\nu (n - \nu)}& \totherwise
	\end{cases}$. This is again an antitonic map, following the same
	argumentation as for Min-Edge-Expansion'. The validator function becomes
	trivial for Sparsest-Cut', that is $\valid(x) \defeq \true$.
	Recovering the Sparsest-Cut' from $\Phi$ works the same way as before:
	it can be expressed as $-f(\Phi)$ which, again, cannot be $\infty$
	unless the graph has at most one vertex.
	\begin{lemma}
	If $n \geq 2$, then $ \displaystyle -f(\Phi)
	= \min_{\emptyset \neq S \subsetneq V}
	\frac{\alpha(\partial S) }{\abs S \cdot \abs {V \setminus S}} $. \stopp
\end{lemma}
\begin{proof}\begin{multline*}
		-f(\Phi) =
		-\max_{\substack{S \subseteq V \\ \valid(\abs S)}}
			f\tup{ \abs S }{ \alpha(\partial S) }
		= \min_{S \subseteq V}
			-f\tup{ \abs S }{ \alpha(\partial S) }
		= \min_{S \subseteq V} \begin{cases}
				\infty & \abs S \in \sing{0, n} \\
				\frac{\alpha(\partial S)}{\abs S \cdot \abs{V \setminus S}} & \totherwise
			\end{cases}
		\\ = \min \Sing{ \infty, \min_{\emptyset \neq S \subseteq V} \frac{\alpha(\partial S)}{\abs S \cdot \abs{V \setminus S}}}
		\overset{n \geq 2}=
		\min_{\emptyset \neq S \subsetneq V} \frac{\alpha(\partial S)}{\abs S \cdot \abs{V \setminus S}}
	\myqedhere\end{multline*}
\end{proof} 
	\para{Densest-Cut'}
	As Densest-Cut' is the complementary problem of Sparsest-Cut', the instantiation
	of the framework is very similar: We set $f (\nu, s) \defeq \begin{cases}
		-\infty & \nu \in \sing{0,n} \\
		\frac{s}{\nu (n - \nu)}& \totherwise
	\end{cases}$ and $\valid(x) \defeq \true$. Then, $f(\Phi)$ is the
	value of the densest cut (and $-\infty$ if there is at most one vertex).
	\begin{lemma}
	If $n \geq 2$, then $ \displaystyle f(\Phi)
	= \max_{\emptyset \neq S \subsetneq V}
	\frac{\alpha(\partial S) }{\abs S \cdot \abs {V \setminus S}} $. \stopp
\end{lemma}
\begin{proof}~\vspace{-1.7em}
	\begin{multline*}
		f(\Phi) =
		\max_{\substack{S \subseteq V \\ \valid(\abs S)}}
			f\tup{ \abs S }{ \alpha(\partial S) }
		= \max_{S \subseteq V}
			f\tup {\abs S} {\alpha(\partial S) }
		= \max_{S \subseteq V} \begin{cases}
				-\infty & \abs S \in \sing{0, n} \\
				\frac{\alpha(\partial S)}{\abs S \cdot \abs{V \setminus S}} & \totherwise
			\end{cases}
		\\ = \max \Sing{ -\infty, \max_{\emptyset \neq S \subseteq V} \frac{\alpha(\partial S)}{\abs S \cdot \abs{V \setminus S}}}
		\overset{n \geq 2}=
		\max_{\emptyset \neq S \subsetneq V} \frac{\alpha(\partial S)}{\abs S \cdot \abs{V \setminus S}}
	\myqedhere \end{multline*}
\end{proof}
  	}

	\onlysubm{
		\section{Omitted Proofs}
		\renewcommand\temptext{Proof of \autoref{stmt:rewrite}}
		\ifllncs{\begin{proof}}{\begin{proof}[\temptext]}
	\begin{multline*}
		\bigsqcup_{\substack{
			S \subseteq X_r \\
			0 \leq \ell \leq \abs {F_r} \\
			\valid(\ell + \abs S)
		}} \B_r(\ell, S)
	=
		\bigsqcup_{\substack{
			S \subseteq X_r \\
			0 \leq \ell \leq \abs {F_r} \\
			\valid(\ell + \abs S)
		}} \bigsqcup_{
		\hat S \in \binom {F_r} \ell 
	} \tup{\ell + \abs S}{\alpha_r (\partial (S \uplus \hat S)}
	=
	\bigsqcup_{\substack{
			S \subseteq X_r \\
			0 \leq \ell \leq \abs {F_r} \\
			\valid(\ell + \abs S) \\
		\hat S \in \binom {F_r} \ell 
	}} \tup{\ell + \abs S}{\alpha_r (\partial (S \uplus \hat S)}
	\\ =
	\bigsqcup_{\substack{
		S \subseteq X_r \\ \hat S \subseteq F_r
	}} \tup{\abs S + \abs {\hat S}} {\alpha_r(\partial(S \uplus \hat S))}
	= \bigsqcup_{\substack{S \subseteq V \\ \valid(\abs S)}}
	\tup {\abs S} {\alpha_r(\partial S)}
	= \Phi \qedhere \end{multline*}
\end{proof}%
 	}

	\section{Correctness of the Recurrences}
	In order to prove the correctness of our recurrences, we first need a few calculation rules
	for the interaction of $\partial A$ and the intersection with specific Cartesian products.
	This will be done in the following section.
	Afterwards, there is one section per node type
	where the correctness of the recurrence for this
	node type is shown formally.

	\subsection{Calculation Rules}
	We start with the simplest and very obvious rule that
states that the weight of edges of the disjunct union
$A \uplus B$
of two disjointedge sets $A, B$ is exactly the sum of the weights of $A$ plus those of $B$:
\begin{lemma} \label{rule:alpha}
	For two sets disjoint edge sets $A$, $B$ we have
	\[ \alpha(A \uplus B) = \alpha(A) + \alpha(B).
	\stopphier \]
\end{lemma}
\begin{proof} Follows directly from definition. \placeqed \end{proof}
Let us now consider some vertex set $A$ and the edges
leaving $A$, that is the edges going from $A$ to $\overline A$. It is not hard to see that the intersection
with $B \times C$ for $B \supseteq A$ and $C \supseteq \overline A$ leaves this set unchanged:
\begin{lemma} \label{stmt:expand} \label{stmt:reduce}
	For any $A$, $B \supseteq A$, and $C \supseteq \overline A$ we have
	\[
		\partial A = \partial A \cap (B \times C). \stopphier
	\]
\end{lemma}
\begin{proof}\begin{multline*}
		\partial A
		= \set{v_1v_2 \in E | v_1 \in A, v_2 \in \overline A}
		= \set{v_1v_2 \in E | v_1 \in A, v_2 \in \overline A} \cap (B \times C)
		\\ = \partial A \cap (B \times C) \qedhere
	\end{multline*}
\end{proof}Now we show rules that are more complex. The following set of rules consider
some edge set $J$ and vertex sets $K$ and $L$.
More precisly, the expression $J \cap (K \times L)$ can
be decomposed if there are certian relations between
the sets involved.
\begin{lemma} \label{l4}
Let $A$ and $B$ be disjoint sets.
\begin{enumerate}[{\normalfont (i)}]
	\item \label{l4:two} \dpm{
		\partial (A \uplus B)
		=     \partial A \cap (A \times \overline B)
		\uplus \partial B \cap (B \times \overline A)
	}
	\suspend
\end{enumerate}
Now let $C$ be a set disjoint with both, $A$ and $B$.
\begin{enumerate}[{\normalfont (i)}]
	\resume
	\item \label{l4:three:right} \dpm{
		\partial A \cap (A \times (B \uplus C))
		= \partial A \cap (A \times B) \uplus \partial A \cap (A \times C)
	}
	\item \label{l4:three:left} \dpm{
		\partial (A \uplus B) \cap ((A \uplus B) \times C)
		= \partial A \cap (A \times C)
		\uplus \partial B \cap (B \times C)
	} \stopp
	\end{enumerate}
\end{lemma}
\begin{proof}
	\subpoint{\stmtref{l4:two}}
\begin{align*}
		\partial(A \uplus B)
		  &=   \set{ v_1v_2 \in E| v_1 \in A \uplus B, v_2 \in \overline{A \uplus B}}
		\\&=   \set{ v_1v_2 \in E| v_1 \in A, v_2 \in \overline{A \uplus B}}
		\uplus \set{ v_1v_2 \in E | v_1 \in B, v_2 \in \overline{A \uplus B} }
		\\&=   \set{ v_1v_2 \in E| v_1 \in A, v_2 \in \overline{A} \land v_2 \in \overline B}
		\uplus \set{ v_1v_2 \in E| v_1 \in B, v_2 \in \overline{A} \land v_2 \in \overline B}
		\\&=   \set{ v_1v_2 \in E | v_1 \in A, v_2 \in \overline A} \cap (A \times \overline B)
		\uplus \set{ v_1v_2 \in E | v_1 \in B, v_2 \in \overline B} \cap (B \times \overline A)
		\\&=   \partial A \cap (A \times \overline B)
		\uplus \partial B \cap (B \times \overline A) 
		\localqedhere
	\end{align*}

	\subpoint{\stmtref{l4:three:right}}
	Note note that $B \uplus C$ and $A$ are disjoint, which implies $B,C \subseteq (B \uplus C) \subseteq \overline A$.
\begin{align*}
		\partial A \cap (A \times (B \uplus C))
		&= \set{v_1v_2 \in E | v_1 \in A, v_2 \in \overline A}
		\cap (A \times (B \uplus C))
		\\
		&\overann={(B \uplus C) \subseteq \overline A} \set{v_1v_2 \in E | v_1 \in A, v_2 \in B \uplus C}
		\\
		&= \set{v_1v_2 \in E | v_1 \in A, v_2 \in B}
		\uplus
		\set{v_1v_2 \in E | v_1 \in A, v_2 \in C}
		\\
		&\overann={B \subseteq \overline A \land C \subseteq \overline A} \partial A \cap (A \times B)
		\uplus \partial A \cap (A \times C) \localqedhere
	\end{align*}

	\subpoint{\stmtref{l4:three:left}}
	Note note that $A \uplus B$ and $C$ are disjoint, which implies $(A \uplus B) \subseteq \overline C$,
	which is equivalent to $\overline {A \uplus B} \supseteq C$.
	\begin{align*}
		\partial (A \uplus B) \cap ((A \uplus B) \times C)
		&= \set{ v_1v_2 \in E | v_1 \in A \uplus B, v_2 \in \overline{A \uplus B}}
		\cap ((A \uplus B) \times C)
		\\
		&\overann={\overline{A \uplus B} \supseteq C}
		\set{v_1v_2 \in E | v_1 \in A \uplus B, v_2 \in C}
		\\
		&= \set{v_1v_2 \in E | v_1 \in A, v_2 \in C}
		\uplus \set{v_1v_2 \in E | v_1 \in B, v_2 \in C}
		\\
		&= \partial A \cap (A \times C) \uplus \partial B \cap (B \times C) \qedhere 
	\end{align*}
\end{proof}
We need those rules for the correctness proof of the join nodes as well as for showing another rule in the following 
paragraph.

For the special case that $B$ contains exactly one element in the first rule of the previous statement, we get the following rule by incorporating \autoref{stmt:reduce}, too:
\begin{corollary} \label{l2}
	For all $v$ and all $S$ where $v \notin S$ it holds that
	\[
		\partial(S \uplus \sing v)
		= \partial S \cap \overline{\sing v}^2
		\uplus \partials v \cap \overline S^2
	.\stopphier \]
\end{corollary}
\begin{proof}
	\begin{align*}
		\partial(S \uplus \sing v)
		&\overann={\autodsubref{l4}{two}}
		\partial S \cap (S \times \overline{\sing v})
		\uplus
		\partials v \cap (\sing v \times \overline S)
		\\
		&\underann={\text{\autoref{stmt:expand}, using $S \subseteq \overline {\sing v}$
		and $\sing v \subseteq \overline S$}}
		\partial S \cap \overline{\sing v}^2
		\uplus
		\partials v \cap \overline S^2 \qedhere
	\end{align*}
\end{proof}
This rule is very important as our recurrences are usually
(except join nodes) built by inductively adding vertices
to the selection set $S$.  
	\subsection{Leaf Node} 
	\begin{lemma}[{Correctness of \autoref{eq:leaf:1}}]
	\label{stmt:correctness:eq:leaf:1}
	Let $i$ be a leaf node and $v \in X_i$. Then for all $S \subseteq X_i$ we have
	\[
		\Gamma_i(0, \sing v) \equiv
		\tup 1 {\sum_{\substack{v' \in X_i \\ vv' \in E}} w(vv')}.
		\tag*{(\ref*{eq:leaf:1})} 
	\]
\end{lemma}
\begin{proof}
We can simply plug in the definition of $\Gamma$ and make use of $Y_i = X_i$ as follows:
	\begin{multline*}
		\Gamma_i(0, \sing v)
		= \bigsqcup_{ \hat S \in \binom {F_i} 0}
			\tup{\abs {\sing v}}{\alpha_i (\partial (\sing v \uplus \hat S))}
		=	\tup1{\alpha_i (\partials v)}
		\\ =   \tup1{\alpha(\partials v \cap Y_i^2)}
		\underann={Y_i = X_i} 
		    \tup1{\alpha(\partials v \cap X_i^2)}  
		= 	\sum_{\substack{v' \in X_i \\ vv' \in E}} w(vv') \qedhere
	\end{multline*}
\end{proof}

\begin{lemma}[{Correctness of \autoref{eq:leaf:2}}]
	\label{stmt:correctness:eq:leaf:2}
	Let $i$ be a leaf node and $v \in X_i$. Then for all $S \subseteq X_i$ and all $\emptyset \neq S' \subseteq X_i \setminus \sing v$  we have
		\[ \Gamma_i(0, S' \uplus \sing v) \equiv \Gamma_i(0, S')
				\oplus \tup1 {\sum_{\substack{v' \in X_i \setminus S' \\ vv' \in E}} w(vv')
				- \sum_{\substack{v' \in S' \\ v'v \in E}} w(v'v) }
			\tag*{(\ref*{eq:leaf:2})} 
		\]
\end{lemma}
\begin{proof}
	We start by plugging in the definition of $\Gamma$:
\begin{equation} \label{proof:correctness:leaf:ii:1}
		\Gamma_i(0, S' \uplus \sing v)
		= \bigsqcup_{
		\hat S \in \binom {F_i} 0
		} \tup{\abs {S' \uplus \sing v}}{\alpha_i \enpar{\partial \enpar{S' \uplus \sing v}}}
		= \tup{1 + \abs {S'}}{\alpha_i\enpar{\partial\enpar{S' \uplus \sing v}}}
	\end{equation}
	The question is now: how can we show that the size of the cut changes (in comparison to the size of the
	cut when only $S'$ is selected instead of $S' \uplus \sing v$)
	as claimed in the statement, that is
	\[	
		\alpha_i(\partial(S' \uplus \sing v)))
		=
		\sum_{\substack{v' \in X_i \setminus S' \\ vv' \in E}} w(vv')
		- \sum_{\substack{v' \in S' \\ v'v \in E}} w(v'v)?
	\]
	To answer this, we explicitly consider the change
	of the cut, that is
	\begin{align*}
		\alpha_i(\partial(S' \uplus \sing v)) - \alpha_i(\partial S')
		&\overann={\autodref{l2}}
		\alpha(\partial S' \cap (X_i\setminus \sing v)^2
		 \uplus \partials v \cap (X_i \setminus S')^2)
		- \alpha_i(\partial S')
\\
		&\overann={\autodref{rule:alpha}}
		\alpha(\partial S' \cap (X_i\setminus \sing v)^2)
		+ \alpha(\partials v \cap (X_i \setminus S')^2)
		- \alpha_i(\partial S')
		\\
		&\overann={\autodref{l2}}
		\alpha(\partial S' \cap (X_i\setminus \sing v)^2)
		+ \alpha(\partials v \cap (X_i \setminus S')^2)
		\\&\phantom{{}={}}
		- \alpha(\partial S' \cap (X_i \times (X_i \setminus \sing v)) \uplus
		\partial S' \cap (X_i \times \sing v))
\\
		&\overann={\autodref{rule:alpha}}
		\alpha(\partial S' \cap (X_i\setminus \sing v)^2)
		+ \alpha(\partials v \cap (X_i \setminus S')^2)
		\\&\phantom{{}={}}
		- \alpha(\partial S' \cap (X_i \times (X_i \setminus \sing v)))
		- \alpha(\partial S' \cap (X_i \times \sing v))
		\\ &\overann={
			v \in X_i \setminus S'
			\komma S' \subseteq X_i
			\komma \autodref{stmt:reduce}}
		\alpha(\partial S' \cap (X_i \times (X_i\setminus \sing v)))
		+ \alpha(\partials v \cap (X_i \setminus S')^2)
		\\&\phantom{{}={}}
		- \alpha(\partial S' \cap (X_i \times (X_i \setminus \sing v)))
		- \alpha(\partial S' \cap (X_i \times \sing v))
\\ &=
		\alpha(\partials v \cap (X_i \setminus S')^2)
		- \alpha(\partial S' \cap (X_i \times \sing v))
\\ &= \sum_{\substack{v' \in X_i \setminus S' \\ vv' \in E}} w(vv')
		- \sum_{\substack{v' \in S' \\ v'v \in E}} w(v'v).
	\end{align*}
	We can use this and \autoref{proof:correctness:leaf:ii:1} to
	deduce
	\begin{align*}
		\Gamma_i(0, S' \uplus \sing v)
		&\overann={\autodref{proof:correctness:leaf:ii:1}} \tup{1+\abs {S'}}{\alpha_i(\partial(S' \uplus \sing v))}
		\\ &= \tup{1 + \abs {S'}}{\alpha_i(\partial(S' \uplus \sing v)) - \alpha_i(\partial S') + \alpha_i(\partial S')}
		\\ &= \tup{\abs {S'}}{\alpha_i(\partial S')}
		\oplus
		\tup1{\alpha_i(\partial(S' \uplus \sing v)) - \alpha_i(\partial S')}
		\\ &\overann={\text{Statement above}}
		\tup{\abs {S'}}{\alpha_i(\partial S')}
		\oplus \tup1
		{\sum_{\substack{v' \in X_i \setminus S' \\ vv' \in E}} w(vv')
		- \sum_{\substack{v' \in S' \\ v'v \in E}} w(v'v)}
		\\ &=
		\Gamma_i(0, S')
				\oplus \tup1 {\sum_{\substack{v' \in X_i \setminus S' \\ vv' \in E}} w(vv')
				- \sum_{\substack{v' \in S' \\ v'v \in E}} w(v'v) }
		\qedhere
	\end{align*}
\end{proof} 
	\subsection{Forget Node}
	\begin{lemma}[Correctness of \autoref{eq:forget}]
Let $i$ be a forget node with child $j$. Let $v \in V$ such that
$X_i \uplus \sing v = X_j$. 
Then for all $S \subseteq X_i$ and all $0 \leq \ell \leq \abs {F_i}$ it holds
that 
\[
	\Gamma_i(\ell, S) \equiv \begin{cases}
		\Gamma_j(\ell-1, S \uplus \sing v) \sqcup \Gamma_j(\ell, S) & \ell \geq 1 \\
		\Gamma_j(\ell, S) & \ell = 0
		\end{cases}.
		\tag*{(\ref*{eq:forget})}
\]
\end{lemma}
\begin{proof}
	Fix $S \subseteq X_i$ and $0 \leq \ell \leq \abs {F_i}$.
	\begin{align*}
		\Gamma_i(\ell, S)
		  &= \bigsqcup_{
			\hat S \in \binom {F_i} \ell 
		} \tup {\ell + \abs S} {\alpha_i (\partial (S \uplus \hat S))}
		\\&\overann={F_i = F_j \uplus \sing v}
		\bigsqcup_{
			\hat S \in \binom {F_j \uplus \sing v} \ell 
		} \tup {\ell + \abs S} {\alpha_i (\partial (S \uplus \hat S))}
		\intertext{Now we make a case distinction wether $v \in \hat S$; if $v \in \hat S$ we may select $\ell -1$ further vertices from $F_j \setminus \sing v$, otherwise we may select $\ell$ vertices from $F_j \setminus \sing v$. This allows us to split the supremum into the supremum of two suprema:}
		&= \left( \bigsqcup_{
			\hat S \in \binom {F_j} {\ell-1} 
		} 
			\tup {\ell + \abs S} {\alpha_i (\partial (S \uplus \hat S \uplus \sing v))}
		\right) \sqcup \left(
		\bigsqcup_{
			\hat S \in \binom {F_j} \ell 
		} \tup{\ell + \abs S} {\alpha_i (\partial (S \uplus \hat S))}
		\right)
		\intertext{Note that the left supremum is applied on an empty set if $\ell = 0$ (and thus vanishes in that case, as $\bigsqcup \emptyset = \bot$); incorporating this and using the definition of $\Gamma_j$ we eventuall get}
		&= \begin{cases}
			\B_j(\ell-1, S \uplus \sing v) \sqcup \B_j(\ell, S) & \ell \geq 1 \\
			\B_j(\ell, S) & \ell = 0
		\end{cases}.
	\qedhere\end{align*}
\end{proof} 
	\subsection{Introduce Node}
	\begin{lemma}[Correctness of \autoref{eq:introduce:1}]
	\label{stmt:eq:introduce:1}
	Let $i$ be an introduce node with child $j$. Let $v \in V$ such that
	$X_i = X_j \uplus \sing v$. Then for all $S \subseteq X_i$ where $v \in S$
	and all $0 \leq \ell \leq \abs{F_i}$ it holds that 
	\begin{equation*}
	\Gamma_i(\ell, S) \equiv
		\tup1{\alpha_i(\partials v \cap (\sing v \times (X_i \setminus S)))}
		\oplus
		\Gamma_j(\ell, S \smin v)
		\tag{\ref*{eq:introduce:1}}
	.\end{equation*}
\end{lemma}
\begin{proof}
Fix $S \subseteq X_i$ and $0 \leq \ell \leq \abs {F_i}$. By the definition of $\Gamma$ we have
	\begin{align*}
		\Gamma_i(\ell, S)
		&= \bigsqcup_{\hat S \in \binom{F_i}{\ell}}
		\tup{\ell + \abs{S}}{\alpha_i(\partial(S \uplus \hat S))}
	\end{align*}
	We want to focus on the $\alpha_i(\dots)$ expression in the equation above.
	Therefor, let $S' \defeq S \setminus \sing v$.
	Consider the
	term $\alpha_i(\partial(S \uplus \hat S))$.
	For the sake of shortness, let $\tilde S \defeq S' \uplus \hat S$.
We then have
	\begin{align*}
		\alpha_i(\partial(\tilde S \uplus \sing v))
		&\overann={\autodref{l2}} \alpha_i(\partial \tilde S \cap \overline{\sing v}^2
		\uplus \partials v \cap \overline{\tilde S}^2)
		\\ &\overann={\autodref{rule:alpha}}
		\alpha_i(\partial \tilde S \cap \overline{\sing v}^2)
		+ \alpha_i(\partials v \cap \overline {\tilde S}^2)
		\\
		&=
		\alpha(\partial \tilde S \cap \overline{\sing v}^2 \cap Y_i^2)
		+
		\alpha(\partials v \cap \overline {\tilde S}^2 \cap Y_i^2)
		\\
		&\overann={\overline {\sing v} = Y_i \setminus \sing v \komma \overline {\tilde S} = Y_i \setminus \tilde S}
		\alpha(\partial \tilde S \cap (Y_i \setminus \sing v)^2)
		+ \alpha(\partials v \cap (Y_i \setminus \tilde S)^2) 
		\intertext{Now note that $v$ has just been introduced; due to
		\subref{defi:tree-decomposition}{edge_coverage} and
		\subref{defi:tree-decomposition}{coherence}
		there cannot be an edge from $v$ to any vertex in $F_i$, thus, using
		$Y_i = X_i \uplus F_i$, we have}
		&= \alpha(\partial \tilde S \cap (X_i \setminus \sing v)^2)
		+ \alpha(\partials v \cap (X_i \setminus \tilde S)^2)
		\\
		&\overann={X_i = X_j \uplus \sing v}
		\alpha(\partial \tilde S \cap X_j^2)
		+ \alpha(\partials v \cap (X_i \setminus \tilde S)^2)
		\\
		&\overann={\tilde S \subseteq X_i \komma \autodref{stmt:reduce}}
		\alpha_j(\partial \tilde S) + \alpha(\partials v \cap (\sing v \times (X_i \setminus \tilde S)))
\\
&\overann={\tilde S = S' \uplus \hat S}
		\alpha_j(\partial \tilde S)
		+ \alpha(\partials v \cap (\sing v \times (X_i \setminus (S' \uplus \hat S))))
		\\
		&\overann={\hat S \subseteq F_i \implies \hat S \cap X_i = \emptyset}
		\alpha_j(\partial \tilde S)
		+ \alpha(\partials v \cap (\sing v \times (X_i \setminus S')))
		\\
		&\overann={vv \notin E \komma \autodref{stmt:reduce} \komma S' = S \setminus \sing v}
		\alpha_j(\partial \tilde S)
		+ \alpha(\partials v \cap (\sing v \times (X_i \setminus S)))
	\end{align*}
	Hence, we can rewrite $\Gamma_i(\ell, S)$ as follows:
	\begin{align*}
		\Gamma_i(\ell, S)
		&= \bigsqcup_{\hat S \in \binom{F_i}{\ell}}
		\tup{\underbrace{\ell + \abs{S}}_{=\mathrlap{\;1 + \ell + \abs{S'}}}}{\alpha_i(\partial(S \uplus \hat S))}
		\\
		&\overann={F_i = F_j}
		\bigsqcup_{\hat S \in \binom{F_j}{\ell}}
		\tup{1 + \ell + \abs{S'}}{\alpha_i(\partial (S' \uplus \hat S))} 
		\\
		&\overann={\text{Statement above}}
		\bigsqcup_{\hat S \in \binom{F_j}{\ell}}
		\tup{1 + \ell + \abs{S'}}{\alpha_j(\partial (S' \uplus \hat S))
		+ \alpha(\partials v \cap (\sing v \times (X_i \setminus S')))} 
		\\
		&=
		\bigsqcup_{\hat S \in \binom{F_j}{\ell}}
		\enpar*{
			\tup{1}{\alpha(\partials v \cap (\sing v \times (X_i \setminus S')))}
			\oplus
			\tup{\ell + \abs{S'}}{\alpha_j(\partial (S' \uplus \hat S))}
		}
		\\
		&\overann={\autodref{stmt:compatible}}
		\tup{1}{\alpha(\partials v \cap (\sing v \times (X_i \setminus S')))}
		\oplus
		\bigsqcup_{\hat S \in \binom{F_j}{\ell}}
		\tup{\ell + \abs{S'}}{\alpha_j(\partial (S' \uplus \hat S))}
		\\&=
		\tup{1}{\alpha(\partials v \cap (\sing v \times (X_i \setminus S')))}
		\oplus \Gamma_j(\ell, S')
		\\&\overann={S' = S \setminus \sing v}
		\tup{1}{\alpha(\partials v \cap (\sing v \times (X_i \setminus S')))}
		\oplus \Gamma_j(\ell, S \setminus \sing v) \qedhere
	\end{align*}
\end{proof} 	\begin{lemma}[Correctness of \autoref{eq:introduce:2}]
	\label{stmt:eq:introduce:2}
	Let $i$ be an introduce node with child $j$. Let $v \in V$ such that
	$X_i = X_j \uplus \sing v$. Then for all $S \subseteq X_i$ where $v \notin S$
	and all $0 \leq \ell \leq \abs{F_i}$ it holds that 
	\begin{equation*}
		\Gamma_i(\ell ,S) \equiv 
		\tup{0}{\alpha_i(\partial S \cap (X_i\times \sing v))}
		\oplus \Gamma_j(\ell, S)
		\tag{\ref*{eq:introduce:2}}
	.\end{equation*}
\end{lemma}
\begin{proof}
	Fix $S \subseteq X_i$ such that $v \in S$, and $0 \leq \ell \leq \abs {F_i}$. By the definition of $\Gamma$ we have
	\begin{align*}
		\Gamma_i(\ell, S)
		&= \bigsqcup_{\hat S \in \binom{F_i}{\ell}}
		\tup{\ell + \abs{S}}{\alpha_i(\partial(S \uplus \hat S))}.
	\end{align*}
	We want to focus on the $\alpha_i(\dots)$ expression. Note that $\alpha_i(\partial(S \uplus \hat S)) = \alpha(\partial \tilde S \cap Y_i^2)$ where $\tilde S \defeq S \uplus \hat S$, thus we will
	have a closer look at the term
	$\tilde S \cap Y_i^2$.
	\begin{align*}
		\partial \tilde S \cap Y_i^2 
		&\overann={\autodref{stmt:reduce}} \partial \tilde S \cap (\tilde S \times (Y_i \setminus \tilde S))
		\\&\overann={Y_i = Y_j \uplus \sing v} \partial \tilde S \cap (\tilde S \times ((Y_j \uplus \sing v) \setminus \tilde S))
		\\& \overann={v \notin \tilde S}
		\partial \tilde S \cap (\tilde S \times ((Y_j \setminus \tilde S) \uplus \sing v))
		\\&= \partial \tilde S \cap (\tilde S \times (Y_j \setminus \tilde S))
		\uplus \partial \tilde S \cap (\tilde S \times \sing v)
		\\&= \partial \tilde S \cap (\tilde S \times (Y_j \setminus \tilde S))
		\uplus \partial \tilde S \cap (X_i \times \sing v)
		\\&\overann={\tilde S \subseteq Y_j \komma \autodref{stmt:reduce}} \partial \tilde S \cap Y_j^2 
		\uplus \partial \tilde S \cap (X_i \times \sing v).
	\end{align*}
	Now focus on the second term of the union above:
	\begin{align*}
		\partial \tilde S \cap (X_i \times \sing v)
		&\overann={\tilde S = S \uplus \hat S} \partial (S \uplus \hat S) \cap (X_i \times \sing v)
		\\&\overann={\autodsubref{l4}{two}} (\partial S \cap (S \times \overline {\hat S})
		\uplus \partial \hat S \cap (\hat S \times \overline S))
		\cap (X_i \times \sing v)
		\\&=
		\partial S \cap (S \times \overline{\hat S}) \cap (X_i \times \sing v)
		\uplus
		\partial \hat S \cap (\hat S \times \overline S) \cap (X_i \times \sing v)
		\\&=
		\partial S \cap ((\underbrace{S \cap X_i}_{=\mathrlap{\;S}}) \times (\underbrace{\overline {\hat S} \cap \sing v}_{= \mathrlap{\;\sing v}}) )
		\uplus
		\partial \hat S \cap ((\underbrace{\hat S \cap X_i}_{= \mathrlap{\;\emptyset}}) \times (\overline S \cap \sing v))
		\\ &=
\partial S \cap (S \times \sing v)
		\\ &\overann={S \subseteq X_i \komma \autodref{stmt:reduce}}
		\partial S \cap (X_i \times \sing v).
	\end{align*}
	Combing both just shown equalities we get that
	\[
		\partial \tilde S \cap Y_i^2
		= \partial \tilde S \cap Y_j^2
		\uplus \partial S \cap (X_i \times \sing v).
	\]
	We can thus deduce, using \autoref{rule:alpha}, that
	\[
		\alpha_i(\partial \tilde S)
		= \alpha_j(\partial \tilde S)
		+ \alpha(\partial S \cap (X_i \times \sing v)).
	\] 
	This allows us to rewrite $\Gamma_i(\ell, S)$ as follows:
	\begin{align*}
		\Gamma_i(\ell, S)
		&= \bigsqcup_{\hat S \in \binom{F_i}{\ell}}
		\tup{\ell + \abs{S}}{\alpha_i(\partial(S \uplus \hat S))}
		\\
		&\overann={\text{Statement above}} \bigsqcup_{\hat S \in \binom{F_i}{\ell}}
		\tup{\ell + \abs{S}}{\alpha_j(\partial (S \uplus \hat S))
		+ \alpha(\partial S \cap (X_i \times \sing v))}
		\\&= \bigsqcup_{\hat S \in \binom{F_i}{\ell}}
		\enpar*{
		\tup{\ell + \abs{S}}{\alpha_j(\partial (S \uplus \hat S))}
		\oplus \tup0{\alpha(\partial S \cap (X_i \times \sing v))}
		}
		\\&\overann={\text{\autoref{stmt:compatible}}}
		\tup0{\alpha(\partial S \cap (X_i \times \sing v))}
		\oplus
		\bigsqcup_{\hat S \in \binom{F_i}{\ell}}
		\tup{\ell + \abs{S}}{\alpha_j(\partial (S \uplus \hat S))}
		\\&\overann={F_i = F_j}
		\tup0{\alpha(\partial S \cap (X_i \times \sing v))}
		\oplus
		\bigsqcup_{\hat S \in \binom{F_j}{\ell}}
		\tup{\ell + \abs{S}}{\alpha_j(\partial (S \uplus \hat S))}
		\\&= 	\tup0{\alpha(\partial S \cap (X_i \times \sing v))}
		\oplus \Gamma_j(\ell, S) \qedhere
	\end{align*}
\end{proof} 	\begin{lemma}[Correctness of \autoref{eq:introduce:3}]
	\label{stmt:eq:introduce:3}
	Let $i$ be an introduce node with child $j$. Let $v \in V$ such that
	$X_i = X_j \uplus \sing v$. Then for all $S \subseteq X_i$
	and all $0 \leq \ell \leq \abs{F_i}$ it holds that 
	\begin{equation*}
		\pi_2 \tuplike{\Gamma_i(\ell, S) \ominus \Gamma_j(\ell, S \cap X_j)}
		= \begin{dcases}
			\alpha_i(\partials v \cap (\sing v \times (X_i \setminus S))) & v \in S \\
			\alpha_i(\partial S \cap (X_i\times \sing v)) & v \notin S 
		\end{dcases}.
		\tag{\ref*{eq:introduce:3}}
	\end{equation*}
\end{lemma}
\begin{proof}
	Fix $S \subseteq X_i$ and $0 \leq \ell \leq \abs{F_i}$.
	If $v \in S$, then by \autoref{stmt:eq:introduce:1}
	\[
		\Gamma_i(\ell, S) \equiv
		\tup1{\alpha_i(\partials v \cap (\sing v \times (X_i \setminus S)))}
		\oplus
		\Gamma_j(\ell, S \smin v)
	\]
	and hence
	\begin{align*}
		&\pi_2\tuplike{\Gamma_i(\ell, S) \ominus \Gamma_j(\ell, S \cap X_j)}
		\\ &= \pi_2\enpar*{\tup1{\alpha_i(\partials v \cap (\sing v \times (X_i \setminus S)))}
		\oplus \Gamma_j(\ell, S \setminus \sing v)
		\ominus \Gamma_j(\ell, S \cap X_j)}
		\\ &= \pi_2\tup1{\alpha_i(\partials v \cap (\sing v \times (X_i \setminus S)))}
		\oplus \pi_2\tuplike{\Gamma_j(\ell, S \setminus \sing v)}
		\ominus \pi_2\tuplike{\Gamma_j(\ell, S \cap X_j)}
		\\ &\overann={v \in S \implies S \cap X_j = S \setminus \sing v}
		\pi_2\tup1{\alpha_i(\partials v \cap (\sing v \times (X_i \setminus S)))}
		\oplus \pi_2\tuplike{\Gamma_j(\ell, S \setminus \sing v)}
		\ominus \pi_2\tuplike{\Gamma_j(\ell, S \setminus \sing v)}
		\\ &=
		\pi_2\tup1{\alpha_i(\partials v \cap (\sing v \times (X_i \setminus S)))}
		\\ &= \alpha_i(\partials v \cap (\sing v \times (X_i \setminus S))).
	\end{align*}
	Now let $v \notin S$. Then, by \autoref{stmt:eq:introduce:2},
	\begin{align*}
		&\pi_2\tuplike{\Gamma_i(\ell, S) \ominus \Gamma_j(\ell, S \cap X_j)}
		\\ &= \pi_2\enpar*{\tup{0}{\alpha_i(\partial S \cap (X_i\times \sing v))}
		\oplus \Gamma_j(\ell, S )
		\ominus \Gamma_j(\ell, S \cap X_j)}
		\\ &\overann={v \notin S \implies S \cap X_j = S} \pi_2\enpar*{\tup{0}{\alpha_i(\partial S \cap (X_i\times \sing v))}
		\oplus \Gamma_j(\ell, S )
		\ominus \Gamma_j(\ell, S)}
		\\ &= \pi_2{\tup{0}{\alpha_i(\partial S \cap (X_i\times \sing v))}}
		\oplus \pi_2\tuplike{\Gamma_j(\ell, S )}
		\ominus \pi_2\tuplike{\Gamma_j(\ell, S)}
		\\ &= 
		\pi_2{\tup{0}{\alpha_i(\partial S \cap (X_i\times \sing v))}}
		\\ &= \alpha_i(\partial S \cap (X_i\times \sing v)). \qedhere
	\end{align*}
\end{proof} 
	\subsection{Join Node}
	In order to prove the correctness for join nodes,
we first need an auxiliary statement.
The next proposition is meant to deal
with the aforementioned \enquote{doubly counted edges}
when building the maximum. Recall that for a join node we claimed that
\[
	\Gamma_i(\ell, S) \equiv \bigg(\bigsqcup_{\substack{
			0 \leq \ell_1 \leq \abs{F_j} \\
			0 \leq \ell_2 \leq \abs{F_k} \\
			\ell_2 + \ell_1 = \ell
		}}
		\tuplike{
			\Gamma_j(\ell_1, S)
		\oplus
			\Gamma_k(\ell_2, S)
		}
		\bigg) \ominus \tup{ \abs S }{ \alpha_i(\partial S) }
		\tag*{(\ref*{eq:join:1})}
	.
\]
When considering the special case of Max-Bisection, this
simplifies to
	\[
		\Gamma_i(\ell, S) \equiv \bigg(\max_{\substack{
			0 \leq \ell_1 \leq \abs{F_j} \\
			0 \leq \ell_2 \leq \abs{F_k} \\
			\ell_2 + \ell_1 = \ell
		}}
		\tuplike{
			\Gamma_j(\ell_1, S)
		+
			\Gamma_k(\ell_2, S)
		}
		\bigg) - \alpha_i(\partial S).
	\]
	The intuition behind the subtraction is somewhat clear: If $\Gamma_j(\ell_1, S)$ is the largest possible cut in $G[Y_j]$ using $S$ and further $\ell_1$ vertices, and $\Gamma_k(\ell_2, S)$ is the largest possible cut in $G[Y_k]$ using $S$ and further $\ell_2$ vertices, their sum includes some summands twice:
	The term $\alpha_i(S)$. This is because all other edges in the largest cut have at least on endpoint that is only in $Y_j$ or only in $Y_k$ (outside $X_i$). The following proposition now shows this formally;
	we need it for the proof of the correctness of the
	equation for join nodes in the \emph{general} case
	as well as for the specific case for Max-Bisection.

\begin{proposition} \label{stmt:join:propo}
	Let $i$ be a join node with left child $j$ and right child $j$.
	Then for all $S \subseteq X_i$, $\hat S_1 \subseteq F_j$,
	and $\hat S_2 \subseteq F_k$ we have
	\[
		\alpha_i(\partial(S \uplus \hat S_1 \uplus \hat S_2))
		= \alpha_j(\partial (S \uplus \hat S_1))
		+ \alpha_k(\partial (S \uplus \hat S_2))
		- \alpha_i(\partial S \cap X_i^2) .\stopp
	\]
\end{proposition}
\begin{proof} Fix $S \subseteq X_i$, $\hat S_1 \subseteq F_j$, and $\hat S_2 \subseteq F_k$.
	We claim that it suffices to show that
	\begin{enumerate}[(i)]
		\item Show that $\partial(S \uplus \hat S_1 \uplus \hat S_2) \cap Y_i^2 = \partial(S \uplus \hat S_1) \cap Y_j^2 \cup \partial(S \uplus \hat S_2) \cap Y_k^2$; and
		\item $\partial(S \uplus \hat S_1) \cap Y_j^2 \cap \partial(S \uplus \hat S_2) \cap Y_k^2
		= \partial S \cap X_i^2$.
	\end{enumerate}
	To see this, assume that both statements hold. We then have
	\begin{align*}
		\,\,\,\,&\!\!\!\!\alpha_i(\partial S \uplus \hat S_1 \uplus \hat S_2)
		\\&=
		\sum_{e \in (\partial S \uplus \hat S_1 \uplus \hat S_2) \cap Y_i^2} w(e)
		\\ &\overann={\text{(i)}}
		\sum_{e \in \partial(S \uplus \hat S_1) \cap Y_j^2 \cup \partial(S \uplus \hat S_2) \cap Y_k^2} w(e)
		\\&=
		\sum_{e \in \partial(S \uplus \hat S_1) \cap Y_j^2} w(e)
		+
		\sum_{e \in \partial(S \uplus \hat S_2) \cap Y_k^2} w(e)
		-
		\sum_{e \in \partial(S \uplus \hat S_1) \cap Y_j^2 \cap \partial(S \uplus \hat S_2) \cap Y_k^2} w(e)
		\\&= \alpha(\partial(S \uplus \hat S_1) \cap Y_j^2) + \alpha(\in \partial(S \uplus \hat S_2) \cap Y_k^2)
		- \alpha(\partial(S \uplus \hat S_1) \cap Y_j^2 \cap \partial(S \uplus \hat S_2) \cap Y_k^2)
		\\&\overann={\text{(ii)}}
		\alpha(\partial(S \uplus \hat S_1) \cap Y_j^2) + \alpha(\partial(S \uplus \hat S_2) \cap Y_k^2)
		- \alpha(\partial S \cap X_i^)
		\\&=\alpha_j(\partial (S \uplus \hat S_1))
		+ \alpha_k(\partial (S \uplus \hat S_2))
		- \alpha_i(\partial S \cap X_i^2).
	\end{align*}
	In order to prove (i) and (ii), we approach as follows: We simplify $\partial(S \uplus \hat S_1 \uplus \hat S_2) \cap Y_i^2$, 
	$\partial(S \uplus \hat S_1) \cap Y_j^2$ and $\partial(S \uplus \hat S_2) \cap Y_k^2$ each to a disjunct union of simple terms. We the compare those \enquote{summands} to prove our goals.
	We start with the left hand side of the equation in (i), that is, the term
	$\partial (S \uplus \hat S_1 \uplus \hat S_2) \cap Y_i^2$:
\begin{align*}
	\,\,\,\,&\!\!\!\!\partial (S \uplus \hat S_1 \uplus \hat S_2) \cap Y_i^2
	\\&\overann={\autodref{stmt:reduce}} \partial (S \uplus \hat S_1 \uplus \hat S_2)
	\cap \big(
	(S \uplus \hat S_1 \uplus \hat S_2) \times (Y_i \setminus (S \uplus \hat S_1 \uplus \hat S_2))
	\big)
	\\
	&\overann={\text{Def $Y_i$}} \partial (S \uplus \hat S_1 \uplus \hat S_2)
	\cap
	\big(
	(S \uplus \hat S_1 \uplus \hat S_2) \times ((X_i \uplus F_i)) \setminus (S \uplus \hat S_1 \uplus \hat S_2))
	\big)
	\\
	&\overann={\text{$i$ is join node}} \partial (S \uplus \hat S_1 \uplus \hat S_2)
	\cap \big(
	(S \uplus \hat S_1 \uplus \hat S_2) \times ((X_i \uplus F_j \uplus F_k)) \setminus (S \uplus \hat S_1 \uplus \hat S_2))
	\big)
	\\
	&\overann={S \subseteq X_i \komma \hat S_1 \subseteq F_j \komma \hat S_2 \subseteq F_k} \partial (S \uplus \hat S_1 \uplus \hat S_2)
	\cap \big(
	(S \uplus \hat S_1 \uplus \hat S_2) \times ((X_i \setminus S) \uplus (F_j \setminus \hat S_1) \uplus (F_k \setminus \hat S_2))
	\big)
	\\
	&\overann={{\autodsubref{l4}{three:left} \komma \autodsubref{l4}{three:right}}}
	\partial S \cap (S \times (X_i \setminus S))
	\uplus \partial S \cap (S \times (F_j \setminus \hat S_1))
	\uplus \partial S \cap (S \times (F_k \setminus \hat S_2))
	\\&\phantom{{}={}}
	\uplus \partial \hat S_1 \cap (\hat S_1 \times (X_i \setminus S))
	\uplus \partial \hat S_1 \cap (\hat S_1 \times (F_j \setminus \hat S_1))
	\uplus \underbrace{
		\partial \hat S_1 \cap (\hat S_1 \times (F_k \setminus \hat S_2))
	}_{=\emptyset}
	\\&\phantom{{}={}}
	\uplus \partial \hat S_2 \cap (\hat S_2 \times (X_i \setminus S))
	\uplus \underbrace{
		\partial \hat S_2 \cap (\hat S_2 \times (F_j \setminus \hat S_1))
	}_{=\emptyset}
	{}\uplus{} \partial \hat S_2 \cap (\hat S_2 \times (F_k \setminus \hat S_2)) 
	\tag{H1}
	\intertext{Let us now consider the first expression on the right side of the equation of (i):}
		\,\,\,\,&\!\!\!\!\partial (S \uplus \hat S_1) \cap Y_j^2
		\\&\overann={\autodref{stmt:reduce}}
		\partial (S \uplus \hat S_1) \cap
		\big(
		(S \uplus \hat S_1) \times (Y_j \setminus (S \uplus \hat S_1))
		\big)
		\\
		&=
		\partial (S \uplus \hat S_1) \cap
		\big(
		(S \uplus \hat S_1) \times ((F_j \uplus X_j) \setminus (S \uplus \hat S_1))
		\big)
		\\
		&\overann={X_i = X_j}
		\partial (S \uplus \hat S_1) \cap
		\big(
		(S \uplus \hat S_1) \times ((F_j \uplus X_i) \setminus (S \uplus \hat S_1))
		\big)
		\\
		&\overann={S \subseteq X_i \komma \hat S_1 \subseteq F_j}
		\partial (S \uplus \hat S_1) \cap
		\big(
		(S \uplus \hat S_1) \times ((X_i \setminus S) \uplus (F_j \setminus \hat S_1))
		\big)
		\\
		&\overann={{\autodsubref{l4}{three:left} \komma \autodsubref{l4}{three:right}}}
		       \partial S \cap (S \times (X_i \setminus S))
		\uplus \partial S \cap (S \times (F_j \setminus \hat S_1))
		\\&\phantom{{}={}}
		\uplus \partial \hat S_1 \cap (\hat S_1 \times (X_i \setminus S))
		\uplus \partial \hat S_1 \cap (\hat S_1 \times (F_j \setminus \hat S_1)) \tag{H2}
	\intertext{Analogously, we can also show}
	\,\,\,\,&\!\!\!\!\partial (S \uplus \hat S_1) \cap Y_k^2
	\\&=
		       \partial S \cap (S \times (X_i \setminus S))
		\uplus \partial S \cap (S \times (F_k \setminus \hat S_2))
		\\&\phantom{{}={}}
		\uplus \partial \hat S_2 \cap (\hat S_2 \times (X_i \setminus S))
		\uplus \partial \hat S_2 \cap (\hat S_2 \times (F_k \setminus \hat S_2))
		\tag{H3}
	\end{align*}

	Now observe that each \enquote{summand} (that is not the empty set) in (H1) appears in (H2) and (H3), and vice versa. This proves (i).

	For (ii), let us compare the \enquote{summands} of
	(H2) and (H3). First observe, that each pair of those has to be identical or disjunct. This is due to the fact
	that all \enquote{summands} of (H2) and (H3) also
	appear in (H1), as we observed above, and (H1) is a
	disjunct union. It is now not hard to see which pair of summands are identical and which are disjunct: Any \enquote{summand} starting with $\partial \hat S_1$ (recall that $\hat S_1 \subseteq F_j)$ in (H2) cannot occur in (H3)
	due to coherence (unless the \enquote{summand} is $\emptyset$), and analogously, no \enquote{summand} starting
	with $\partial \hat S_2$ (recall that $\hat S_2 \subseteq F_k$) in (H3) can occur in (H2). Moreover, the same applies to \enquote{summands} in (H2) that are intersected with a subset of $V \times F_j$;
	and analogously, to \enquote{summands} in (H3) that are intersected with a subset of $V \times F_k$.
	The only \enquote{summand} left is $\partial S \cap (S \times (X_i \setminus S))$, and this summand occurs (in this very syntactical form) in both, (H2)
	and (H3). This shows (ii) and hence eventually the claim. \placeqed
\end{proof}
\begin{lemma}[Correctness of \autoref{eq:join:1}]
Let $i$ be a join node with left child $j$ and right child $j$.
Then for all $0 \leq \ell \leq \abs {F_i}$ and all $S \subseteq X_i$ we have
\[ 
	\Gamma_i(\ell, S) \equiv \bigg(\bigsqcup_{\substack{
		0 \leq \ell_1 \leq \abs{F_j} \\
		0 \leq \ell_2 \leq \abs{F_k} \\
		\ell_2 + \ell_1 = \ell
	}}
	\tuplike{
		\Gamma_j(\ell_1, S)
	\oplus
		\Gamma_k(\ell_2, S)
	}
	\bigg) \ominus \tup{\abs S}{\alpha_i(\partial S)}.
	\tag*{(\ref*{eq:join:1})}
\]
\end{lemma}
\begin{proof} Fix $S \subseteq X_i$ and $0 \leq \ell \leq \abs{F_i}$. By the definition of $\Gamma$ we have:
	\begin{align*} 
		\,\,\,\,&\!\!\!\! \B_i(\ell, S) 
		  \\ &= \bigsqcup_{
			\hat S \in \binom {F_i} \ell 
		} \tup {\ell + \abs S} {\alpha_i (\partial (S \uplus \hat S))}
		\\&= \bigsqcup_{\substack{
			0 \leq \ell_1 \leq \ell \\
			\ell_2 = \ell - \ell_1 \\
			\hat S_1 \in \binom {F_j} {\ell_1} \\
			\hat S_2 \in \binom {F_k} {\ell_2}
		}} \tup {\ell + \abs S} {\alpha_i (\partial (S \uplus \hat S_1 \uplus \hat S_2))}
		\\&\overann={\autodref{stmt:join:propo}} \bigsqcup_{\substack{
			0 \leq \ell_1 \leq \ell \\
			\ell_2 = \ell - \ell_1 \\
			\hat S_1 \in \binom {F_j} {\ell_1} \\
			\hat S_2 \in \binom {F_k} {\ell_2}
		}} \tup{\ell + \abs S} {\alpha_j (\partial (S \uplus \hat S_1)) + \alpha_k ( \partial( S \uplus \hat S_2)) - \alpha_i(\partial S) }
		\\&=
		\bigsqcup_{\substack{
			0 \leq \ell_1 \leq \ell \\
			\ell_2 = \ell - \ell_1
		}}
		\bigsqcup_{\hat S_1 \in \binom {F_j} {\ell_1}}
		\bigsqcup_{\hat S_2 \in \binom {F_k} {\ell_2}}
			\tup{\ell + \abs S} {\alpha_j (\partial (S \uplus \hat S_1))
			+ \alpha_k(\partial(S \uplus \hat S_2))
			- \alpha_i(\partial S)
			}
		\\&\overann={\autodref{stmt:compatible}}
		\phantom{x\,} \left(
		\bigsqcup_{\substack{
			0 \leq \ell_1 \leq \ell \\
			\ell_2 = \ell - \ell_1
		}}
		\bigsqcup_{\hat S_1 \in \binom {F_j} {\ell_1}}
		\bigsqcup_{\hat S_2 \in \binom {F_k} {\ell_2}}
			\tup {\ell + 2\abs S} {\alpha_j (\partial (S \uplus \hat S_1))
			+ \alpha_k (\partial (S \uplus \hat S_2)) } \right)
\\&\phantom{{}={}}\ominus \tup {\abs S} {\alpha_i(\partial S)}
		\\&\overann={\autodref{stmt:compatible}}
		\phantom{x\,}
		\left(
		\bigsqcup_{\substack{
			0 \leq \ell_1 \leq \ell \\
			\ell_2 = \ell - \ell_1
		}}
		\bigsqcup_{\hat S_1 \in \binom {F_j} {\ell_1}}
		\bigsqcup_{\hat S_2 \in \binom {F_k} {\ell_2}}
		\!\!\!\!\enpar*{
			\tup {\ell_1 + \abs S} {\alpha_j (\partial (S \uplus \hat S_1))}
			\oplus \tup{\ell_2 + \abs S}{\alpha_k (\partial (S \uplus \hat S_2))} } \!\!\right)\hskip -2em
		\\&\phantom{{}={}}\ominus \tup {\abs S} {\alpha_i(\partial S)}
		\\&=\left(\bigsqcup_{\substack{
			\abs S \leq \ell_1 \leq \ell \\
			\ell_2 = \ell - \ell_1
		}}
		\enpar*{
			\bigsqcup_{\hat S_1 \in \binom {F_j} {\ell_1}}
			\!\!\!\!\tup{
				\ell_1 + \abs S }{ \alpha_j (\partial (S \uplus \hat S_1))  
			} \oplus
			\bigsqcup_{\hat S_2 \in \binom {F_k} {\ell_2}}
				\!\!\!\!\tup {\ell_2 + \abs S} {\alpha_k (\partial (S \uplus \hat S_2))}
		\!\!} \!\!\right)\hskip -2em
		\\&\phantom{{}={}}\ominus \tup {\abs S} {\alpha_i(\partial S)}
		\\&=\left(\bigsqcup_{\substack{
			0 \leq \ell_1 \leq \ell \\
			\ell_2 = \ell - \ell_1
		}}
		\tuplike{
			\B_j(\ell_1, S)
		\oplus
			\B_k(\ell_2, S)
		}
		\right)
		\ominus \tup {\abs S} {\alpha_i(\partial S)}. \qedhere
	\end{align*}
\end{proof}
 
	\section{Running Time Of Our Framework}
	\onlysubm{We claim that our algorithm has an overall running time of
	$\O(2^tn^2)$. To prove this we first need to deal with some preprocessing
	steps. The first step was already mentioned: converting the tree decomposition
	into a nice tree decomposition. This can be done in $\O(nt^2)$ according to
	\autoref{stmt:nice}. We model sets as bit vectors of length $n$; we omit the
	(very technical) details on how set operations can be implemented on bit vectors
	such that they all take only constant time.
	The second preprocessing step is creating an adjacency matrix of the graph
	in time $\O(n^2)$; this is necessary to be able to compute the explicit sums
	occurring in our recurrences in time $\O(n)$ each. We now claim that our dynamic
	program takes the running time specified in \autoref{figure:running_time_goals}
	per node, depending on the node's type. As we only have $\O(n)$ nodes in the
	decomposition (it is a small decomposition), this directly implies that
	the computation of $\Gamma_i$ only takes time $\O(2^tn^2)$.

	\begin{figure}[t]
		\centering
		\begin{tabulary}{\textwidth}{|L||LL|}
			\hline
			Leaf node & $\O(2^t n)$ & per node 
			\\ \hline
			Forget node & $\O(2^t \abs{F_i})$ & per node
			\\ \hline
			Introduce node & $\O(2^t n)$ & per node
			\\ \hline
			\multirow{2}{*}{Join node} & $\O(2^t \abs{F_j \times F_k})$ & per node
			\\ 
			& $\O(2^t n^2)$ & for all nodes
			\\ \hline
		\end{tabulary}
		\caption{Running times of the computation of $\Gamma_i$ for
		a node $i \in I$ (if there is one child $i$, let $j$ be that child, and if
		there are two children, let $j,k$ be those children of $i$) depending
		on the node's type}
		\label{figure:running_time_goals}
	\end{figure}

	\paras{Leaf Node} Let $i$ be a leaf node. Then $F_i = \emptyset$.
	Hence $\Gamma_i$ has $\O(2^t)$ entries. We make use of \multiautoref{eq:leaf:1, eq:leaf:2}
	to compute these entries; both equations have a constant number of sums, and
	each sum can be computed in $\O(n)$ using the adjacency matrix and basic
	set operations.

	\paras{Forget Node}
	Let $i$ be a forget node with child $j$. Then $X_i \uplus \sing v = X_j$
	for some $v \in V$. There are $\O(2^t \abs{F_i})$ entries in $\Gamma_i$.
	As we only use \autoref{eq:forget} to compute the values, which takes
	constant time per value, the overall running time for a single forget node
	is $\O(2^t \abs{F_i})$.
	
	\paras{Introduce Node}
	Let $i$ be an introduce node with child $j$. Then $X_i = X_j \uplus \sing v$
	for some $v \in V$. There are $\O(2^t \abs{F_i}) \subseteq \O(2^tn)$ entries
	in $\Gamma_i$ we need to compute. For the entries $\Gamma_i(0, S)$ for
	$S \subseteq X_i$ we use \multiautoref{eq:introduce:1, eq:introduce:2} and explicitly compute the
	$\alpha_i$ expressions as sums in $\O(n)$ each. There are at most
	$\O(2^t)$ entries for $\ell = 0$, thus for all these entries we need at most
	$\O(2^tn)$ time.
	For all other $\O(2^tn - 2^t) = \O(2^tn)$ entries we make use
	of \autoref{eq:introduce:3} which allows us now to use 
	\multiautoref{eq:introduce:1, eq:introduce:2} to compute the remaining entries
	in time $\O(1)$ each. The overall running time for a single introduce node
	is hence $\O(2^tn)$.

	\paras{Join Node}
	Let $i$ be a join node with left child $j$ and right child $k$.
	Then $X_i = X_j = X_k$. There are $\O(2^t \abs{F_i}) \subseteq \O(2^t \abs{F_j}\cdot\abs{F_k})$
	entries we need
	to compute. We only use \autoref{eq:join:1} to compute them. For a fixed $S \subseteq
	X_i$, we need to compare $\abs{F_j} \cdot \abs{F_k}$ different pairs in that equation.
	This means
	that the running time -- per entry -- is in $\O(\abs{F_j} \cdot \abs{F_k})
	= \O(\abs{F_j \times F_k})$. Using \autoref{eq:nn} (which is a consequence of
	\autoref{stmt:pairs}) it follows that the overall running time for all entries
	of all join nodes is bounded by $\O(2^tn^2)$. }

	\section{Hardness}
	We need the following Proposition in order to be able to show that the
	hardness result of Max-Cut by \textcite{DBLP:journals/talg/LokshtanovMS18}
	can be extended to the case where a tree decomposition is given as part
	of the input. The idea is to compute such a decomposition for the result
	of their reduction from SAT to Max-Cut in polynomial time.
	\begin{propo} \label{lemma:k-forest-decomp}
		Let $G = (V,E)$ and $S \subseteq V$.
		If $G - S$ is a forest, then (given $S$)
		one can compute a small tree decomposition of $G$ that has width
		$\leq \abs S + 1$
		in time $\O(\poly(\abs E + \abs V))$. \stopp
	\end{propo}
	\begin{proof}
		Fix $G = (V,E)$ and $S$.
		Note that it is trivially possible to compute a small tree
		decomposition of a tree in time linear in its nodes
		and edges. Compute the tree decompositions for all connected
		components in $G' \defeq G - S$ (which are all trees by assumption).
		Rename
		the node sets (and all the occurrences accordingly)
		of the decompositions such that the union of all node
		sets is of the form $[k] \setminus \set 1$ for some $k \in \N$
		and the node sets are disjunct.
		Add a new node $1$ to the forest
		of tree decompositions and set the associated bag to be
		$X_1 \defeq \emptyset$. Connect the new node $1$ to
		an arbitrary node of each tree decomposition. 
		It is easy to check that we just constructed a
		small tree decomposition for $G'$ of width $1$.
		Now modify this decomposition to a small tree decomposition
		of $G$ as follows: Add the set $S$ to all bags. Again,
		it is easy to check that we obtain a small tree decomposition
		of $G$. The same holds for the running time.
		The width of the constructed decomposition is bounded
		from above by the term $1 + \abs S$ as every bag of the
		decomposition for $G'$ had at most $2$ vertices before
		we added the vertices within $S$.
	\end{proof}
	\begin{lemma}
		\label{thm:adv-max-cut-hard}
		There is no $\O((2-\eps)^{t} \poly n)$ algorithm to decide Max-Cut
		for any $\eps \geq 0$. \stopp
	\end{lemma}
	\begin{proof}
		Consider the graph created in the reduction by \textcite{DBLP:journals/talg/LokshtanovMS18}
		and observe that following statements:
		\begin{enumerate}
			\item \label{app:thm:adv-max-cut-hard:proof:1}
			The number of vertices after the reduction is $\O(\abs \phi + n) = \O(\abs \phi)$;
			\item \label{app:thm:adv-max-cut-hard:proof:2}
			For every clause there is exactly one cycle created;
			\item \label{app:thm:adv-max-cut-hard:proof:3}
			Any pair of such cycles does share exactly one vertex: $x_0$;
			\item \label{app:thm:adv-max-cut-hard:proof:4}
			Apart from those partially overlapping cycles, there are only
			$n$ additional vertices;
			\item \label{app:thm:adv-max-cut-hard:proof:5}
			Every vertex on those cycles that is not $x_0$
			has degree at most $3$ as it is connected to at most
			one $\widehat v_i$;
		\end{enumerate}
From \itemref{app:thm:adv-max-cut-hard:proof:4} and
		\itemref{app:thm:adv-max-cut-hard:proof:5}
		we can now deduce
		that the number of edges after the reduction is
		$\O(\abs \phi)$. Thus the encoding of
		the graph has asymptotically the same size as the encoding
		of $\phi$. The correctness of the reduction was shown in
		\cite{DBLP:journals/talg/LokshtanovMS18}.

		Now observe that \itemref{app:thm:adv-max-cut-hard:proof:3}
		and \itemref{app:thm:adv-max-cut-hard:proof:4} imply that
		removing all variable vertices $\widehat v_i$ and $x_0$
		does result in a set of paths. This is exactly the statement
		we need in order to be able to apply \autoref{lemma:k-forest-decomp}
		which completes the proof.
	\end{proof}
\end{document}